\documentclass[12pt]{article}
\usepackage[english]{babel}
\usepackage{bm}
\usepackage{fullpage}
\usepackage{cite}

\usepackage{newpxtext,newpxmath}

\usepackage[utf8]{inputenc}
\usepackage{amsthm}
\usepackage{amssymb}
\usepackage{amsmath}
\usepackage{bbold}
\usepackage{bbm}
\usepackage[pdftex, backref=page]{hyperref}
\usepackage{braket}
\usepackage{dsfont}
\usepackage{mathdots}
\usepackage{mathtools}
\usepackage{enumerate}
\usepackage[shortlabels]{enumitem}
\usepackage{csquotes}
\usepackage{stmaryrd}
\usepackage[cal=boondox]{mathalfa}
\usepackage{graphicx}
\usepackage{stackengine}
\usepackage{scalerel}
\usepackage{tensor}       %\tensor[_n]{\braket{j|\psi}}{_{1\ldots n}}
\usepackage{array}
\usepackage{makecell}
\newcolumntype{x}[1]{>{\centering\arraybackslash}p{#1}}
\usepackage{tikz}
\usepackage{pgfplots}
\usetikzlibrary{shapes.geometric, shapes.misc, positioning, arrows, arrows.meta, decorations.pathreplacing, decorations.pathmorphing, patterns, angles, quotes, calc}
\usepackage{booktabs}
\usepackage{xfrac}
\usepackage{siunitx}
\usepackage{centernot}
\usepackage{comment}
\usepackage{chngcntr}
\usepackage{caption}
\usepackage{subcaption}

\newcommand{\Sym}{{\rm Sym}}
\newtheorem{thm}{Theorem}
\newtheorem*{thm*}{Theorem}

\newtheorem*{prop*}{Proposition}
\newtheorem{lemma}[thm]{Lemma}
\newtheorem*{lemma*}{Lemma}
\newtheorem{cor}[thm]{Corollary}
\newtheorem*{cor*}{Corollary}

\newtheorem*{cj*}{Conjecture}
\newtheorem{Def}[thm]{Definition}
\newtheorem*{Def*}{Definition}

\newtheorem*{question*}{Question}

\newtheorem*{problem*}{Problem}
\newtheorem*{example*}{Example}

\makeatletter
\def\thmhead@plain#1#2#3{%
  \thmname{#1}\thmnumber{\@ifnotempty{#1}{ }\@upn{#2}}%
  \thmnote{ {\the\thm@notefont#3}}}
\let\thmhead\thmhead@plain
\makeatother

\theoremstyle{definition}
\newtheorem{rem}[thm]{Remark}

\newcommand{\bb}{\begin{equation}\begin{aligned}\hspace{0pt}}
\newcommand{\bbb}{\begin{equation*}\begin{aligned}}
\newcommand{\ee}{\end{aligned}\end{equation}}
\newcommand{\eee}{\end{aligned}\end{equation*}}

\newcommand{\ketbra}[1]{\ket{#1}\!\!\bra{#1}}

\renewcommand{\epsilon}{\varepsilon}

\newcommand{\id}{\mathrm{id}}

\DeclareMathOperator{\Tr}{Tr}

\DeclareMathAlphabet{\pazocal}{OMS}{zplm}{m}{n}

\DeclareMathOperator{\supp}{supp}

\DeclareMathOperator{\Id}{\rm{id}}

\newcommand{\lsmatrix}{\left(\begin{smallmatrix}}
\newcommand{\rsmatrix}{\end{smallmatrix}\right)}

\stackMath

\stackMath

\makeatletter
\newcommand*\rel@kern[1]{\kern#1\dimexpr\macc@kerna}
\newcommand*\widebar[1]{%
  \begingroup
  \def\mathaccent##1##2{%
    \rel@kern{0.8}%
    \overline{\rel@kern{-0.8}\macc@nucleus\rel@kern{0.2}}%
    \rel@kern{-0.2}%
  }%
  \macc@depth\@ne
  \let\math@bgroup\@empty \let\math@egroup\macc@set@skewchar
  \mathsurround\z@ \frozen@everymath{\mathgroup\macc@group\relax}%
  \macc@set@skewchar\relax
  \let\mathaccentV\macc@nested@a
  \macc@nested@a\relax111{#1}%
  \endgroup
}

\counterwithin*{equation}{part}
\counterwithin*{thm}{part}
\counterwithin*{figure}{part}

\tikzset{meter/.append style={draw, inner sep=10, rectangle, font=\vphantom{A}, minimum width=30, line width=.8, path picture={\draw[black] ([shift={(.1,.3)}]path picture bounding box.south west) to[bend left=50] ([shift={(-.1,.3)}]path picture bounding box.south east);\draw[black,-latex] ([shift={(0,.1)}]path picture bounding box.south) -- ([shift={(.3,-.1)}]path picture bounding box.north);}}}
\tikzset{roundnode/.append style={circle, draw=black, fill=gray!20, thick, minimum size=10mm}}
\tikzset{squarenode/.style={rectangle, draw=black, fill=none, thick, minimum size=10mm}}

\definecolor{Blues5seq1}{RGB}{239,243,255}
\definecolor{Blues5seq2}{RGB}{189,215,231}
\definecolor{Blues5seq3}{RGB}{107,174,214}
\definecolor{Blues5seq4}{RGB}{49,130,189}
\definecolor{Blues5seq5}{RGB}{8,81,156}

\definecolor{Greens5seq1}{RGB}{237,248,233}
\definecolor{Greens5seq2}{RGB}{186,228,179}
\definecolor{Greens5seq3}{RGB}{116,196,118}
\definecolor{Greens5seq4}{RGB}{49,163,84}
\definecolor{Greens5seq5}{RGB}{0,109,44}

\definecolor{Reds5seq1}{RGB}{254,229,217}
\definecolor{Reds5seq2}{RGB}{252,174,145}
\definecolor{Reds5seq3}{RGB}{251,106,74}
\definecolor{Reds5seq4}{RGB}{222,45,38}
\definecolor{Reds5seq5}{RGB}{165,15,21}

\allowdisplaybreaks

\usepackage[most,breakable]{tcolorbox}
	{\expandafter\ifstrequal\expandafter{#1}{orange}{\begin{tcolorbox}[colback=red!15,colframe=orange!15,breakable,enhanced]}{\begin{tcolorbox}[colback=Blues5seq1,colframe=Blues5seq5,breakable,enhanced]}}%
	{\end{tcolorbox}}

\usepackage{authblk}
\newcommand{\Ran}{{\rm{Ran}}\,}

\newtheorem*{remark}{Remark}

\allowdisplaybreaks[1] 
\usepackage{graphicx,import,xcolor,transparent}
\usepackage{centernot}

\begin{document}

\title{Robustness of Entanglement Manipulation for almost iid sources\ }
\author{Nilanjana Datta}

\affil{
Department of Applied Mathematics and Theoretical Physics\\
University of Cambridge\\
Cambridge CB3 0WA, United Kingdom
}

\date{}

\maketitle
\begin{abstract}
We study the robustness of asymptotic entanglement manipulation for
Mazzola--Sutter--Renner (MSR) almost-iid sources, which allow a sublinear
number of deviations from tensor-power structure.
For pure sources along \(\ket{\phi}_{AB}\), every concentration rate
below \(S(\phi_A)\) is achievable by a single Schur--Weyl protocol
depending only on the rate and local Hilbert spaces.
For general sources along \(\rho_{AB}\), every distillation rate below
\(E_D(\rho_{AB})\) is achievable by a single LOCC protocol sequence
independent of the particular source and defect size. Its error
vanishes exponentially and uniformly for each prescribed defect bound
\(r_n=o(n)\), yielding an optimal universal distillation rate equal to
\(E_D(\rho_{AB})\).
For dilution, the entanglement cost of every MSR almost-iid target
sequence along \(\rho_{AB}\) is at most \(E_F^\infty(\rho_{AB})\), the
regularized entanglement of formation of the reference state.
Consequently, every MSR almost-iid sequence along a pure reference
state \(\ket{\phi}_{AB}\) is asymptotically reversible, with distillable
entanglement and entanglement cost equal to \(S(\phi_A)\), even when
its individual states are mixed.
The proofs establish structural and entropic properties of MSR
almost-iid sources that may be useful in other information-theoretic
settings.
\end{abstract}

\section{Introduction}

Entanglement is a fundamental non-local resource in quantum information
theory, playing a central role in quantum teleportation, superdense
coding, quantum cryptography and quantum error correction. Since
entanglement monotones cannot increase on average under local
operations and classical communication (LOCC), a central problem in
entanglement theory is to determine the optimal rates at which
bipartite states can be converted to and from maximally entangled
states under LOCC.
In entanglement distillation, many copies of a shared bipartite state are
converted into maximally entangled states; in the special case of pure
states, this task is usually called entanglement concentration. The
reverse task, entanglement dilution, asks for the minimum rate at which
maximally entangled states must be consumed in order to prepare a target
state. The corresponding optimal asymptotic rates are the distillable
entanglement \(E_D\) and the entanglement cost \(E_C\), respectively.

For an iid pure-state source \(\ket{\phi}_{AB}^{\otimes n}\),
entanglement concentration and dilution are both governed by the entropy
of entanglement \(S(\phi_A)\), where
\(\phi_A=\Tr_B\ket{\phi}\!\bra{\phi}_{AB}\)
\cite{BennettBernsteinPopescuSchumacher1996,LoPopescu1999}.
For iid mixed-state sources, the entanglement cost is given by the
regularized entanglement of formation \(E_F^\infty\), whereas the
distillable entanglement is generally harder to characterize and may be
strictly smaller than the entanglement cost
\cite{BennettDiVincenzoSmolinWootters1996,HaydenHorodeckiTerhal2001}.

These tasks have also been studied in more general and refined settings,
including information-spectrum formulations for arbitrary sequences of
bipartite states
\cite{Hayashi2006,BowenDatta2006,bowen2011entanglement,
bowen2008asymptotic,Matsumoto2007},
as well as finite-blocklength and second-order analyses of entanglement
concentration, dilution, and related pure-state conversion problems
\cite{KumagaiHayashi2013,DattaLeditzky2014,
FangWangTomamichelDuan2017}.

The iid assumption is one of the central idealizations in quantum
Shannon theory: asymptotic rates are usually derived under the premise
that the underlying source is exactly a tensor power. Although this
assumption is mathematically powerful, the resulting conclusions need
not remain valid under structural perturbations. In realistic settings,
correlations, memory effects, local imperfections or preparation errors
may lead to sources which approximate tensor-power behaviour only in a
weaker structural or statistical sense. This raises a basic robustness
question: which information-theoretic conclusions survive when the
exact iid hypothesis is replaced by an appropriate notion of almost
iid behaviour? Recent work has shown that the answer depends
crucially on the chosen notion of approximation
\cite{girardi2026new,girardi2026quantum}.

Three related notions of almost iid quantum sources have recently
been studied. The most restrictive is the class introduced by Mazzola,
Sutter and Renner. An MSR almost-iid source admits a
permutation-invariant extension which, at blocklength \(n\), is supported
on the span of permutations of vectors consisting of \(n-r_n\) copies
of a fixed purification of the reference state and an arbitrary vector
on the remaining \(r_n=o(n)\) tensor factors
\cite{{mazzola2026almost}}. Mazzola, Sutter and Renner showed that the
conditional entropy per tensor factor converges to the corresponding
conditional entropy of the reference state, and applied this result to
establish robustness of squashed entanglement.

Girardi, De Palma and Lami introduced two broader classes:
Wasserstein almost iid and weakly almost iid sources
\cite{girardi2026new}. For a Wasserstein almost iid source, the
normalized quantum Wasserstein distance from the reference tensor power
vanishes asymptotically. For a weakly almost iid source, fixed-size
marginals converge, on average, to the corresponding tensor powers of
the reference state. They showed that, in general, MSR almost-iid
sources form a strict subclass of Wasserstein almost iid sources,
which in turn form a strict subclass of weakly almost iid sources.
They also showed that, for pure reference states, the Wasserstein almost
iid and weakly almost iid classes coincide.

For weakly almost iid sources, a noncommutative weak law of large
numbers and a universal entropy-typicality theorem were established in
\cite{datta2026entropy}, yielding, among other consequences, direct
proofs of bounds for data compression and asymmetric hypothesis
testing. The robustness of asymmetric hypothesis testing, classical
and quantum data compression, and classical communication over quantum
channels under different notions of almost iid structure was
investigated in \cite{girardi2026quantum}.

Brand{\~a}o and Eisert considered a worst-case model of correlated
distillation in which the copies may be arbitrarily correlated but have
identical one-copy reductions \cite{BrandaoEisert2008}. They studied the
qualitative question of whether such correlations can prevent
distillation, rather than asymptotic distillation rates; their model
imposes no almost iid or sublinear-deviation condition.

For the special case of non-iid sequences supported on the fully
symmetric subspace of
\((\mathcal H_A\otimes\mathcal H_B)^{\otimes n}\), Matsumoto derived
formulas for the entanglement cost and distillable entanglement
\cite{Matsumoto2007}. Pure MSR almost-iid sources fall within this
symmetric-support setting for sufficiently large blocklengths.
Matsumoto's formulas, however, depend on the full source sequence and do
not by themselves imply that the resulting rates coincide with the
corresponding iid rates of the MSR almost-iid reference state. Moreover,
general mixed MSR almost-iid sources need not have symmetric support,
so his symmetric-source results do not apply to them.

In the present paper, we establish robustness of the standard
iid achievability bounds for entanglement concentration of pure
MSR almost-iid sources, entanglement distillation from arbitrary
MSR almost-iid sources, and entanglement dilution for arbitrary
MSR almost-iid target sequences.

\subsection{Main results}

We establish three main achievability results for MSR almost-iid
sources: universal entanglement concentration for pure sources,
universal entanglement distillation for arbitrary, possibly mixed,
sources, and an entanglement-dilution bound for arbitrary, possibly
mixed, target sequences.

Our first main result concerns entanglement concentration. For an
iid source \(\ket{\phi}_{AB}^{\otimes n}\), the optimal
concentration rate is the entropy of entanglement \(S(\phi_A)\), where
\(
\phi_A=\Tr_B\ket{\phi}\!\bra{\phi}_{AB}.
\)
A natural question is whether this rate remains achievable when the
iid source is replaced by an arbitrary pure MSR almost-iid source along
\(\ket{\phi}_{AB}\); see Definitions~\ref{def:MSR}
and~\ref{def:MSR-source}, and Remark~\ref{rem:pure-msr}. A stronger
question is whether a single sequence of concentration protocols can
be used for every source in the class, without knowledge of the
particular source sequence.

Theorem~\ref{thm:universal-concentration-pi-pure-msr} answers both
questions affirmatively. For any fixed sequence \(r_n=o(n)\), every rate
\(R<S(\phi_A)\) can be achieved by a single sequence of LOCC protocols
whose error tends to zero uniformly over all pure MSR almost-iid sources
along \(\ket{\phi}_{AB}\) having defect size at most \(r_n\) at blocklength
\(n\). The protocols depend only on \(R\) and the local Hilbert spaces,
and not on the particular source sequence; the reference state enters
only through the admissible range \(R<S(\phi_A)\). The main new
ingredient in the proof is a uniform Schur--Weyl concentration estimate
under the prescribed sublinear defect bound. Combining this estimate
with structural and spectral-entropy properties of MSR almost-iid
sources allows us to apply the universal concentration protocol of
Hayashi and Matsumoto~\cite{PhysRevA.75.062338} throughout the class.

{Our second main result concerns universal
entanglement distillation from arbitrary, possibly mixed, MSR almost-iid
sources along a bipartite reference state \(\rho_{AB}\).
Our Theorem~\ref{thm:MSR-universal-distillation} shows that every rate
\(0\leq R<E_D(\rho_{AB})\) is achievable by a single sequence of LOCC
protocols depending only on \(\rho_{AB}\), \(R\), and the local Hilbert
spaces. The protocols are independent of the particular source sequence
and of the defect-size sequence \((r_n)_n\).
For each prescribed sequence \((r_n)_n\) satisfying \(r_n=o(n)\),
their error vanishes exponentially and uniformly over the corresponding
MSR almost-iid source class. In particular, for every MSR almost-iid
source sequence
\(\boldsymbol{\omega}_{AB}:=(\omega_{A^nB^n})_{n\geq1}\)
along \(\rho_{AB}\), with defect sizes \(r_n=o(n)\), we have
\bb
E_D(\boldsymbol{\omega}_{AB})\geq E_D(\rho_{AB}).
\ee
The optimal universal distillation rate for each such class equals
\(E_D(\rho_{AB})\), since the class includes the exact iid source.

The proof bounds each MSR almost-iid input state, in the operator
order, by a scalar multiple of the tensor power of a state formed
by mixing \(\rho_{AB}\) with a small amount of the maximally mixed
state. We use the local robustness construction in the proof of
Lemma~S6 of Tirone et al.~\cite{Tirone26}, together with the
finite-block hashing estimate in their Lemma~S5,
Eqs.~(S29)--(S30), to obtain a single iid distillation protocol
whose error decreases exponentially in the blocklength.
This estimate holds uniformly on \(\sigma_{AB}^{\otimes n}\)
for all bipartite states \(\sigma_{AB}\) in a fixed
trace-distance neighbourhood of \(\rho_{AB}\).
The admixture used in the operator bound is a particular state
in this neighbourhood. Choosing its mixing parameter sufficiently
small ensures that transferring the error estimate through the
operator bound retains exponential decay uniformly over the
MSR almost-iid source class.}

Our third main result concerns the reverse task of entanglement
dilution. Here we consider arbitrary, possibly mixed, MSR almost-iid
target sequences along a bipartite state \(\rho_{AB}\).
Theorem~\ref{thm:MSR-dilution} shows that every dilution rate
\(
R>E_F^\infty(\rho_{AB})
\)
is achievable for any such sequence. This extends the standard
iid achievability bound to the MSR almost-iid class. Importantly, the
states \(\rho_{A^nB^n}\) in an MSR almost-iid target sequence along
\(\rho_{AB}\) need not be close in trace distance to
\(\rho_{AB}^{\otimes n}\), and the sublinear deviations allowed by the
MSR almost-iid condition may be correlated and delocalized across the
tensor factors. Nevertheless, the iid dilution bound remains achievable
for the actual target sequence. For each fixed sequence \(r_n=o(n)\),
the same rate bound holds uniformly over all MSR almost-iid target
sequences having defect size at most \(r_n\) at blocklength \(n\).
This does not assert that a single dilution protocol works for every
target sequence in the class: the protocol may depend on the particular
target sequence.

{Combining universal distillation and dilution yields asymptotic
reversibility for every MSR almost-iid sequence along a pure reference
state, including sequences whose individual states are mixed.
Specifically, Theorem~\ref{thm:pure-msr-dilution-converse} shows that,
for any such sequence \(\bm\omega_{AB}\) along
\(\phi_{AB}:=\ket{\phi}\!\bra{\phi}_{AB}\),
\bb
E_D(\bm\omega_{AB})
=
E_C(\bm\omega_{AB})
=
S(\phi_A).
\ee
This holds for every prescribed sublinear defect bound \(r_n=o(n)\).
The distillation protocol can be chosen independently of the particular
source, whereas the dilution protocol may depend on the prescribed
target sequence. For pure source sequences,
Theorem~\ref{thm:universal-concentration-pi-pure-msr} additionally gives
a concentration protocol independent of the reference state.
For mixed reference states, our results establish the bounds
\bb
E_D(\bm\omega_{AB})\geq E_D(\rho_{AB}),
\qquad
E_C(\bm\omega_{AB})\leq E_F^\infty(\rho_{AB}),
\ee
while the corresponding equalities for individual MSR almost-iid
sequences remain open in general. These results address the robustness questions raised by Mazzola,
Sutter and Renner~\cite[Questions~5.2 and~5.3]{mazzola2026almost}.
As made precise in Corollary~\ref{cor:msr-block-entanglement-bounds},
our distillation and dilution theorems establish, respectively,
the lower bound for 
distillable entanglement per tensor factor
and the upper bound for entanglement cost per tensor factor,
for every prescribed \(r_n=o(n)\).} Both questions are answered affirmatively for pure
reference states, including when the MSR almost-iid block states
are mixed. The reverse inequalities remain open for general mixed
reference states.

The proofs also yield several auxiliary results that may be useful
beyond the entanglement manipulation tasks considered here. We
establish stability of the MSR almost-iid class under local tensor-power
channels, partial traces, and blocking operations, together with
subexponential dimension bounds for the associated defect subspaces.
We also prove entropy-rigidity estimates showing that the relevant
spectral entropy rates of MSR almost-iid sources coincide with the
corresponding entropy quantities of the reference state. For the
dilution theorem, we develop a cq-lifting construction with subexponential
classical overhead and combine it with a blocking argument to relate an
MSR almost-iid target sequence to the ensemble description appearing in
the information-spectrum formula for entanglement cost.

\paragraph{Layout of the paper.}
After introducing the necessary notation and mathematical
preliminaries, we define MSR almost-iid states and sources in
Section~\ref{sec:def-msr}. In Section~\ref{sec:props-msr}, we establish
their structural stability properties and dimension bounds for their
defect subspaces in Section~\ref{subsec:struct}, and their
entropy-rigidity properties in Section~\ref{subsec:entropy-rigid}.
In Section~\ref{sec:ent-manipulate}, we introduce entanglement
distillation, concentration, and dilution for general sequences of
states and define the corresponding operational quantities.
Entanglement concentration for pure MSR almost-iid sources is studied in
Section~\ref{sec:ent-conc}, whose main result is the universal
concentration theorem stated as
Theorem~\ref{thm:universal-concentration-pi-pure-msr}.
Section~\ref{sec:MSR-universal-distillation} establishes universal
entanglement distillation for arbitrary MSR almost-iid sources; see
Theorem~\ref{thm:MSR-universal-distillation}.
Section~\ref{sec:msr-ent-dil} concerns entanglement dilution for MSR
almost-iid target sequences along a fixed bipartite reference state; see
Theorem~\ref{thm:MSR-dilution}. We conclude in
Section~\ref{sec:conclude}. Proofs of several technical lemmas are deferred
to the appendices.

{\begin{remark}[(Comparison with subsequent work)]
After the first version of this work appeared, Mazzola and
Renner~\cite{MazzolaRenner2026Operational} established an operational
framework for almost-iid state transformations. They proved equality
of the operational distillation rate with the iid distillable
entanglement (Corollary~5.16), and equality of their operational
distillation and formation rates with the entanglement entropy for
pure reference states (Corollary~5.17). Their structural
characterization of these rates uses power-law defect bounds
\(r_n=O(n^\kappa)\), with \(0\leq\kappa<1\).

Our Theorem~\ref{thm:MSR-universal-distillation}
applies to every prescribed \(r_n=o(n)\) and gives an explicit
exponential error bound uniform over the corresponding MSR
almost-iid source class. For dilution, their operational cost
optimizes over admissible target sequences (Proposition~5.24),
whereas our Theorem~\ref{thm:MSR-dilution} bounds the cost of every
prescribed MSR almost-iid target sequence. Combining our
distillation and dilution results,
our Theorem~\ref{thm:pure-msr-dilution-converse} establishes asymptotic
reversibility for every MSR almost-iid sequence along a pure
reference state, including sequences whose individual states are
mixed, for every prescribed defect bound \(r_n=o(n)\).
\end{remark}}

\section{Notations and Definitions}\label{sec:not-def}

Let \(\mathcal H\) be a finite-dimensional Hilbert space. We write
\(\mathcal L(\mathcal H)\) for the algebra of linear operators on
\(\mathcal H\), and \(\mathcal D(\mathcal H)\) for the set of density
operators on \(\mathcal H\), i.e.\ positive semidefinite operators of unit
trace. A pure state is a rank-one projection \(\ket{\psi}\!\bra{\psi}\),
where \(\ket{\psi}\in\mathcal H\) is a unit vector. 
For bipartite systems \(A\) and \(B\), and any positive integer
$M\leq\min\{\dim\mathcal H_A,\dim\mathcal H_B\}$, we denote by
$\Phi_M:=\ket{\Phi_M}\!\bra{\Phi_M}
    \in\mathcal D(\mathcal H_A\otimes\mathcal H_B)
$
the standard maximally entangled state of Schmidt rank \(M\), where
$
\ket{\Phi_M}
   :=\frac{1}{\sqrt M}\sum_{i=1}^M
     \ket{i}_A\otimes\ket{i}_B,
$
and \(\{\ket{i}_A\}_{i=1}^M\) and
\(\{\ket{i}_B\}_{i=1}^M\) are orthonormal sets in
\(\mathcal H_A\) and \(\mathcal H_B\), respectively.

{The identity operator on \(\mathcal H\) is denoted by
\(I_{\mathcal H}\), or simply by \(I\) when the underlying
Hilbert space is clear. We label different quantum systems by
capital Roman letters (\(A,B,F,A'\), etc.) and often use these
letters interchangeably with the corresponding Hilbert spaces.}

An operator \(\Pi\in\mathcal L(\mathcal H)\) is an orthogonal projection
if \(\Pi=\Pi^\dagger=\Pi^2\). Its range is denoted by
\(\Ran(\Pi)\). For a positive semidefinite operator \(X\), we write
\(\supp X:=\Ran(\{X>0\})\), where \(\{X>0\}\) denotes the spectral
projection onto the strictly positive part of the spectrum of \(X\). More
generally, if \(Q=\sum_i\lambda_i\ket{\psi_i}\!\bra{\psi_i}\) is
self-adjoint and \(r\in\mathbb R\), we use the notation
\(\{Q>r\}:=\sum_{\lambda_i>r}\ket{\psi_i}\!\bra{\psi_i}\), and similarly
for \(\{Q\le r\}\), \(\{Q\ge r\}\), and \(\{Q<r\}\).

For \(p\in[1,\infty)\), the Schatten \(p\)-norm of
\(X\in\mathcal L(\mathcal H)\) is
\(\|X\|_p:=(\Tr |X|^p)^{1/p}\), where
\(|X|:=\sqrt{X^\dagger X}\). The Schatten \(\infty\)-norm is the operator
norm, denoted by \(\|X\|_\infty\). In particular,
\(\|X\|_1=\Tr |X|\) is the trace norm and
\(\|X\|_2=(\Tr X^\dagger X)^{1/2}\) is the Hilbert--Schmidt norm.

The von Neumann entropy of \(\rho\in\mathcal D(\mathcal H)\) is
\(S(\rho):=-\Tr\rho\log\rho\). Throughout the paper, all logarithms are
taken to base \(2\). For a bipartite state $\rho_{AB} \in {\mathcal D}({\mathcal H}_A \otimes {\mathcal H}_B)$, the conditional entropy is given by $S(A|B)_\rho = S(\rho_{AB}) - S(\rho_B)$, where $\rho_B:= \Tr_A \rho_{AB}$ is the reduced state of the system $B$. For \(x\in[0,1]\), we denote the binary entropy by
$
h(x):=-x\log x-(1-x)\log(1-x),
$
with the convention \(0\log0:=0\).

A quantum channel from \(\mathcal H\) to \(\mathcal K\) is a completely
positive trace-preserving linear map
\(\pazocal E:\mathcal L(\mathcal H)\to\mathcal L(\mathcal K)\). We shall
also use the same symbol for its restriction to states. The identity channel
is denoted by \(\id\).

The fidelity between two states \(\rho,\sigma\in\mathcal D(\mathcal H)\)
is defined by
\bb\label{eq:fidelity}
F(\rho,\sigma)
:=
\|\sqrt{\rho}\sqrt{\sigma}\|_1 .
\ee
With this convention, \(0\le F(\rho,\sigma)\le1\), and
\(F(\rho,\sigma)=1\) if and only if \(\rho=\sigma\). The purified distance
is
\bb\label{eq:purified}
P(\rho,\sigma)
:=
\sqrt{1-F(\rho,\sigma)^2}.
\ee
{The trace distance between \(\rho,\sigma\in\mathcal D(\mathcal H)\) is defined as 
\bb\label{eq:trace-distance}
T(\rho,\sigma):=\frac12\|\rho-\sigma\|_1.
\ee
The Fuchs--van de Graaf inequalities~\cite{FuchsVanDeGraaf1999} then take the form
\bb
1-F(\rho,\sigma)
\leq T(\rho,\sigma)
\leq P(\rho,\sigma).
\ee}

We next recall the min-entropy quantities~\cite{Renner2005, Tomamichel2016} used below. For a state
\(\rho_A\in\mathcal D(\mathcal H_A)\), the min-entropy is
\bb
H_{\min}(A)_\rho
:=
-\log\|\rho_A\|_\infty .
\ee
Equivalently, \(H_{\min}(A)_\rho\) is the largest real number \(\lambda\)
such that \(\rho_A\le 2^{-\lambda}I_A\). 
For a state \(\rho\), we also use the notation \(H_{\min}(\rho):=-\log\|\rho\|_\infty\).

For a bipartite state
\(\rho_{AB}\), the conditional min-entropy is
\bb
H_{\min}(A|B)_\rho
:=
-\inf_{\sigma_B\in\mathcal D(\mathcal H_B)}
D_{\max}\!\left(
\rho_{AB}\middle\|I_A\otimes\sigma_B
\right),
\ee
where
\bb
D_{\max}(\omega\|\tau)
:=
\inf\{\lambda\in\mathbb R:\omega\le 2^\lambda\tau\}
\ee
whenever \(\supp\omega\subseteq\supp\tau\), and is \(+\infty\) otherwise~\cite{Datta08, Renner2005}.
Equivalently, \(H_{\min}(A|B)_\rho\) is the largest real number
\(\lambda\) for which there exists a state \(\sigma_B\) such that
\(\rho_{AB}\le 2^{-\lambda} I_A\otimes\sigma_B\).

For \(\rho\in\mathcal D(\mathcal H)\) and \(\varepsilon,\delta\ge0\), we
write
\bb
B_{\mathrm P}^{\varepsilon}(\rho)
:=
\{\sigma\in\mathcal D(\mathcal H):P(\sigma,\rho)\le\varepsilon\},
\qquad
B_1^\delta(\rho)
:=
\{\sigma\in\mathcal D(\mathcal H):\|\sigma-\rho\|_1\le\delta\}.
\ee
The \(\varepsilon\)-smooth min-entropy with respect to purified distance is
\bb
H_{\min}^{\varepsilon,\mathrm P}(A)_\rho
:=
\sup_{\sigma_A\in B_{\mathrm P}^{\varepsilon}(\rho_A)}
H_{\min}(A)_\sigma,
\ee
and the trace-norm smoothed version is defined analogously by optimizing
over \(B_1^\delta(\rho_A)\)\footnote{That is, \(H_{\min}^{\delta,\|\cdot\|_1}(A)_\rho
:=
\sup_{\sigma_A\in B_1^\delta(\rho_A)}
H_{\min}(A)_\sigma\).}. Similarly, for a bipartite state
\(\rho_{AB}\), the purified-distance smooth conditional min-entropy is
\bb
H_{\min}^{\varepsilon,\mathrm P}(A|B)_\rho
:=
\sup_{\sigma_{AB}\in B_{\mathrm P}^{\varepsilon}(\rho_{AB})}
H_{\min}(A|B)_\sigma,
\ee
with the trace-norm version obtained by replacing
\(B_{\mathrm P}^{\varepsilon}(\rho_{AB})\) by \(B_1^\delta(\rho_{AB})\).

By the Fuchs--van de Graaf inequalities~\cite{FuchsVanDeGraaf1999},
\(\|\rho-\sigma\|_1\le2P(\rho,\sigma)\). Hence
\bb\label{eq:ball-inclusion}
B_{\mathrm P}^{\varepsilon}(\rho)
\subseteq
B_1^{2\varepsilon}(\rho).
\ee
Consequently,
\bb\label{eq:h-min}
H_{\min}^{2\varepsilon,\|\cdot\|_1}(A)_\rho
\ge
H_{\min}^{\varepsilon,\mathrm P}(A)_\rho,
\qquad
H_{\min}^{2\varepsilon,\|\cdot\|_1}(A|B)_\rho
\ge
H_{\min}^{\varepsilon,\mathrm P}(A|B)_\rho.
\ee
Thus any lower bound on the min-entropy smoothed with respect to the
purified distance immediately yields the corresponding lower bound for
trace-norm smoothing, with smoothing parameter enlarged from
\(\varepsilon\) to \(2\varepsilon\).

We finally recall the key entropic quantities of the quantum information-spectrum framework~\cite{HayashiNagaoka2003, BowenDattaISIT2006} (see also, e.g.~\cite{Hayashi2017} and references therein). Let
\(\bm\rho=(\rho_n)_{n\ge1}\) and
\(\bm\sigma=(\sigma_n)_{n\ge1}\) be sequences of positive semidefinite
operators acting on finite-dimensional Hilbert spaces
\((\mathcal H_n)_{n\ge1}\), with each \(\rho_n\) a state. Following the
formulation of~\cite{bowen2006beyond}, the spectral sup-divergence rate and spectral
inf-divergence rate are defined by
\bb
\overline D(\bm\rho\|\bm\sigma)
:=
\inf\left\{
\gamma\in\mathbb R:
\lim_{n\to\infty}
\Tr\!\left[
\{\rho_n-2^{n\gamma}\sigma_n>0\}\rho_n
\right]
=0
\right\},
\ee
and
\bb
\underline D(\bm\rho\|\bm\sigma)
:=
\sup\left\{
\gamma\in\mathbb R:
\lim_{n\to\infty}
\Tr\!\left[
\{\rho_n-2^{n\gamma}\sigma_n>0\}\rho_n
\right]
=1
\right\}.
\ee

The corresponding spectral entropy rates are obtained by comparing the
state sequence with the identity sequence, i.e.\ if
\(\bm\rho=(\rho_n)_n\) is a sequence of states on
\((\mathcal H_n)_n\), and \(\bm I=(I_n)_n\), where
\(I_n\) is the identity operator on \(\mathcal H_n\), then
\bb\label{eq:sup-spectral-entropy-rate}
\overline S(\bm\rho)
:=
-\underline D(\bm\rho\|\bm I),
\qquad
\underline S(\bm\rho)
:=
-\overline D(\bm\rho\|\bm I).
\ee
We use overlines for spectral sup-rates and underlines for spectral inf-rates.
\begin{rem} Equivalently, \(\overline S(\bm\rho)\) is the infimum over all real
\(R\) for which there exists a sequence of projections \((\Pi_n)_n\)
satisfying \(\Tr(\Pi_n\rho_n)\to1\) as \(n\to\infty\), and
\(\Tr\Pi_n\le 2^{nR}\) for all sufficiently large \(n\). 
(We include an explanation of this remark in Appendix~\ref{app:info-spec} for the convenience of readers who are unfamiliar with the Information Spectrum Framework.)
\end{rem}

We refer to \(\overline S(\bm\rho)\) and
\(\underline S(\bm\rho)\), respectively, as the spectral sup-entropy rate and the spectral inf-entropy rate of the source
\(\bm\rho\). Note that \(\underline S(\bm\rho)\le\overline S(\bm\rho)\).
\smallskip

For a sequence of bipartite states
\(\bm\rho_{AB}=(\rho_{A^nB^n})_n\), the spectral conditional sup- and inf-entropy rates
are, respectively, defined by
\bb\label{eq:cond-sup-spec}
\overline S(A|B)_{\bm\rho}
:=
-\underline D\!\left(
\bm\rho_{AB}
\middle\|
(I_{A^n}\otimes\rho_{B^n})_n
\right),
\ee
and
\bb
\underline S(A|B)_{\bm\rho}
:=
-\overline D\!\left(
\bm\rho_{AB}
\middle\|
(I_{A^n}\otimes\rho_{B^n})_n
\right),
\ee
where \(\rho_{B^n}:=\Tr_{A^n}\rho_{A^nB^n}\). These quantities reduce to
the usual conditional entropy \(S(A|B)_\rho=S(\rho_{AB})-S(\rho_B)\) in
the iid case \(\rho_{A^nB^n}=\rho_{AB}^{\otimes n}\).

\medskip

\subsection{MSR almost-iid states}\label{sec:def-msr}
Let \(S_n\) denote the permutation group on \(\{1,\ldots,n\}\), and let
\({\mathcal D}(\mathcal H)\) be the set of density operators acting on a finite-dimensional Hilbert space \(\mathcal H\).
The symmetric subspace of \(\mathcal H^{\otimes n}\) is defined by
\bb
\mathrm{Sym}^n(\mathcal H)
:=
\operatorname{span}
\{
\ket{\phi}^{\otimes n}:\ket{\phi}\in\mathcal H
\}.
\ee

{Given a unit vector \(\ket{\theta}\in\mathcal H\) and integers
\(0\leq m\leq n\), let}
\bb
\mathcal V(\mathcal H^{\otimes n},\ket{\theta}^{\otimes m})
:=
\left\{
\pi\!\left(
\ket{\theta}^{\otimes m}\otimes\ket{\Omega}
\right)
:
\pi\in S_n,\ 
\ket{\Omega}\in\mathcal H^{\otimes(n-m)}
\right\},
\ee
{that is, the set of vectors obtained by permuting a tensor product
consisting of \(m\) copies of a pure state \(\ket{\theta}\) together with an arbitrary
pure state of the remaining \(n-m\) subsystems. We further define}
\bb
\mathrm{Sym}^n(\mathcal H,\ket{\theta}^{\otimes m})
:=
\mathrm{Sym}^n(\mathcal H)
\cap
\operatorname{span}
\mathcal V(\mathcal H^{\otimes n},\ket{\theta}^{\otimes m}).
\ee
The notion of almost iid states was originally introduced for pure symmetric states: a pure state
\(
\ket{\Psi^{(n)}}\in
\mathrm{Sym}^n(\mathcal H,\ket{\theta}^{\otimes(n-r)})
\)
was called a \(\binom{n}{r}\)-almost iid state in \(\ket{\theta}\).
To accommodate mixed states and more general correlated structures, we use the following broader definition {introduced by Mazzola, Sutter and Renner~\cite{mazzola2026almost}.}

\begin{Def}[(MSR almost-iid states)]\label{def:MSR}
{Let \(\mathcal H_A\) be a finite-dimensional Hilbert space,
let \(\sigma_A\in\mathcal D(\mathcal H_A)\),
and let \(n\in\mathbb N_+\) and \(r\in\{0,\ldots,n\}\).}
A state
\(
\rho_{A_1^n}\in {\mathcal D}(\mathcal H_A^{\otimes n})
\)
is called a \(\binom{n}{r}\)-almost iid state along \(\sigma_A\)
if there exist a purification \(\ket{\theta}_{AE}\) of \(\sigma_A\) and an extension
\(\rho_{A_1^nE_1^n}\) of \(\rho_{A_1^n}\) such that:

\begin{enumerate}
\item[(i)]
\(\rho_{A_1^nE_1^n}\) is invariant under simultaneous permutations of the subsystem pairs \((A_i,E_i)\);

\item[(ii)]
\bb\label{eq:supp}
\operatorname{supp}(\rho_{A_1^nE_1^n})
\subseteq
\operatorname{span}
\mathcal V
\bigl(
\mathcal H_{AE}^{\otimes n},
\ket{\theta}_{AE}^{\otimes(n-r)}
\bigr).
\ee
\end{enumerate}
The set of all such states is denoted by\footnote{We use the notation of Mazzola--Sutter--Renner~\cite{mazzola2026almost} here. For generalizations of the MSR almost-iid source see~~\cite{mazzola2026almost} and \cite{girardi2026new}.}
$
{\mathcal S}^n(\mathcal H_A,\sigma_A^{\otimes(n-r)}).
$

In the following, we shall refer to such an extension $\rho_{A_1^n E_1^n}$ as an 
{\emph{MSR extension}} of the state $\rho_{A_1^n}$ along the purification \(\ket{\theta}_{AE}\) of $\sigma_A$.
\end{Def}
\smallskip

\noindent
Note that $\mathcal V
\bigl(
\mathcal H_{AE}^{\otimes n},
\ket{\theta}_{AE}^{\otimes(n-r)}
\bigr)$ is a set of vectors rather than a linear subspace; the corresponding linear subspace is obtained by taking its span.

\begin{rem}[(Unrestricted tensor positions and defect size)]
The parameter \(r\) measures the number of tensor positions that are not
constrained to carry the reference purification \(\ket{\theta}_{AE}\).
More precisely, each vector in
\(
\mathcal V
\bigl(
\mathcal H_{AE}^{\otimes n},
\ket{\theta}_{AE}^{\otimes(n-r)}
\bigr)
\)
has \(n-r\) tensor positions fixed to be in the state
\(\ket{\theta}_{AE}\), while the remaining \(r\) tensor positions are
arbitrary and may be jointly correlated. We refer to these remaining
positions as \emph{unrestricted tensor positions}, and to \(r\) as the
\emph{defect size}.
\end{rem}

\begin{rem}[(Defect spaces)]\label{rem:defect-spaces}
{For a purification \(\ket{\theta}_{AE}\) of \(\sigma_A\) and an integer
\(0\leq r\leq n\), we define the associated \emph{defect space} by}
\bb\label{eq:defect-space}
\mathcal M_n(\theta,r)
:=
\operatorname{span}
\mathcal V
\bigl(
\mathcal H_{AE}^{\otimes n},
\ket{\theta}_{AE}^{\otimes(n-r)}
\bigr).
\ee
Thus, if
\(
\rho_{A_1^n}\in
\mathcal S^n(\mathcal H_A,\sigma_A^{\otimes(n-r)})
\),
then by Definition~\ref{def:MSR} there exists an MSR extension
\(\rho_{A_1^nE_1^n}\) satisfying
\bb
\operatorname{supp}(\rho_{A_1^nE_1^n})
\subseteq
\mathcal M_n(\theta,r).
\ee
We shall frequently refer to \(\mathcal M_n(\theta,r)\) as the {\emph{defect space}} associated with the reference purification \(\ket{\theta}_{AE}\) and defect parameter \(r\).
\end{rem}

\begin{Def}[(MSR almost-iid sources)]
\label{def:MSR-source}
Let \(\rho_A\in\mathcal D(\mathcal H_A)\), and let
\((r_n)_n\) be a sequence of nonnegative integers satisfying
\(r_n\le n\). A sequence of states
\(
\bm\rho_A:=(\rho_{A^n})_{n\ge1}
\),
with
\(
\rho_{A^n}\in\mathcal D(\mathcal H_A^{\otimes n})
\),
is called an MSR almost-iid source along \(\rho_A\), with defect sizes
\((r_n)_n\), if for every \(n\),
\bb
\rho_{A^n}
\in
\mathcal S^n(\mathcal H_A,\rho_A^{\otimes(n-r_n)}).
\ee
In the asymptotic regime considered in this paper, we shall always assume
\(r_n=o(n)\).
\end{Def}

The same terminology will be used for bipartite sources
\(
\bm\rho_{AB}:=(\rho_{A^nB^n})_n
\)
by applying the preceding definition to the single composite system
\(AB\), with reference state \(\rho_{AB}\).

\begin{rem}[(Pure MSR almost-iid sources)]\label{rem:pure-msr}
A sequence
\(
\bm\Psi_{AB}:=(\ket{\Psi_n}\!\bra{\Psi_n}_{A^nB^n})_{n}
\)
is called a pure MSR almost-iid source along a bipartite pure state
\(\ket{\phi}_{AB}\), with defect sizes \((r_n)_n\), if
\bb\label{eq:pure-pi-msr}
\ket{\Psi_n}_{A^nB^n}
\in
\Sym^n\!\left(
\mathcal H_A\otimes\mathcal H_B,
\ket{\phi}_{AB}^{\otimes(n-r_n)}
\right)
\ee
for every \(n\). In particular, each \(\ket{\Psi_n}\) is invariant under
simultaneous permutations of the \(n\) bipartite tensor factors. 
\end{rem}

\begin{rem}
    (Terminology) For brevity we often refer to an MSR almost-iid sequence of states (or source) along a reference state simply as an MSR sequence (or source) along a reference state.
\end{rem}
\section{Structural and entropic rigidity of MSR almost-iid states}\label{sec:props-msr}

In this section we establish the basic stability and rigidity properties
of MSR almost-iid states which will be used throughout the paper. The
first group of results concerns the behaviour of the MSR almost-iid structure under
natural operations, such as local tensor-power channels, and gives
subexponential bounds on the size of the associated defect spaces. The
second group of results shows that, despite the possible presence of
sublinear deviations from the reference tensor-power structure, MSR almost-iid
sources retain the relevant asymptotic spectral entropy properties of
their reference states.

\subsection{Structural stability and defect spaces}\label{subsec:struct}

\begin{lemma}[(MSR almost-iid stability under local tensor-power channels)]
\label{lem:msr-stability-local-channels}
Let \(\bm\rho_{AB}:=(\rho_{A^nB^n})_n\) be MSR almost-iid along
\(\rho_{AB}\), with defect sizes \(r_n=o(n)\). Let
\(\Lambda_A:\mathcal D(\mathcal H_A)\to\mathcal D(\mathcal H_{A'})\) be a quantum channel, and define
\bb
\omega_{A'^nB^n}
:=
(\Lambda_A^{\otimes n}\otimes\Id_{B^n})(\rho_{A^nB^n}),
\qquad
\omega_{A'B}
:=
(\Lambda_A\otimes\Id_B)(\rho_{AB}).
\ee
Then \(\bm\omega_{A'B}:=(\omega_{A'^nB^n})_n\) is MSR almost-iid along
\(\omega_{A'B}\), with the same defect sizes \(r_n=o(n)\).
\end{lemma}

\begin{proof} See Appendix~\ref{app:structure-msr-stability}.
\end{proof}

\begin{cor}[(MSR almost-iid stability under taking marginals)]
\label{cor:msr-marginals}
Let \(\bm\rho_{AB}:=(\rho_{A^nB^n})_n\) be MSR almost-iid along
\(\rho_{AB}\), with defect sizes \(r_n=o(n)\). Then the marginal source
\(\bm\rho_B:=(\rho_{B^n})_n\), where
\(\rho_{B^n}:=\Tr_{A^n}\rho_{A^nB^n}\), is MSR almost-iid along
\(\rho_B:=\Tr_A\rho_{AB}\), with the same defect sizes.

\end{cor}

\begin{proof}
Apply Lemma~\ref{lem:msr-stability-local-channels} with the channel
\(\Lambda_A=\Tr_A\), whose output system is one-dimensional. Then
\bb
(\Lambda_A^{\otimes n}\otimes\Id_{B^n})(\rho_{A^nB^n})
=
\Tr_{A^n}\rho_{A^nB^n}
=
\rho_{B^n},
\ee
and
\bb
(\Lambda_A\otimes\Id_B)(\rho_{AB})
=
\Tr_A\rho_{AB}
=
\rho_B.
\ee
Hence \(\bm\rho_B\) is MSR almost-iid along \(\rho_B\), with the same defect sizes \(r_n=o(n)\).
\end{proof}

The following lemma quantifies the size of the defect spaces that
appear in Remark~\ref{rem:defect-spaces}. It shows that when the number of defects is sublinear,
the corresponding defect spaces have only subexponential dimension.

\begin{lemma}[(Subexponential defect complexity)]
\label{lem:subexp-defect}
{Let \(\ket{\theta}_{AE}\) be a unit vector in the fixed
finite-dimensional Hilbert space
\(\mathcal H_A\otimes\mathcal H_E\), and let
\((r_n)_n\) be a sequence of integers satisfying
\(0\leq r_n\leq n\) and \(r_n=o(n)\).} Then the corresponding MSR defect spaces
\(\mathcal M_n(\theta,r_n)\), defined in~\eqref{eq:defect-space}, satisfy
\bb
\dim \mathcal M_n(\theta,r_n)=2^{o(n)}.
\ee
Consequently, every vector in \(\mathcal M_n(\theta,r_n)\) can be expanded using at most
\(2^{o(n)}\) basis vectors.
\end{lemma}

\begin{proof}
By Remark~2.4(c) of~\cite{mazzola2026almost},
\bb
\dim \mathcal M_n(\theta,r_n)
\le
\binom{n}{r_n}d^{r_n},
\ee
where \(d=\dim(\mathcal H_A\otimes\mathcal H_E)\). Taking logarithms and
dividing by \(n\), we obtain
\bb
\frac1n\log \dim \mathcal M_n(\theta,r_n)
\le
\frac1n\log\binom{n}{r_n}
+
\frac{r_n}{n}\log d.
\ee
{Since \(r_n=o(n)\), we have \(r_n/n\to0\) as \(n\to\infty\).
Moreover, the standard bound
\(\log\binom{n}{r_n}\leq n h(r_n/n)\), which follows
from the binomial theorem, gives
\bb
\frac1n\log\binom{n}{r_n}
\leq
h\!\left(\frac{r_n}{n}\right)
\longrightarrow0
\qquad\text{as }n\to\infty.
\ee}
Hence,
\bb
\lim_{n\to\infty}
\frac1n\log \dim \mathcal M_n(\theta,r_n)
=0,
\ee
which is equivalent to
\(\dim\mathcal M_n(\theta,r_n)=2^{o(n)}\).

The final assertion follows by choosing an orthonormal basis of
\(\mathcal M_n(\theta,r_n)\).
\end{proof}

\subsection{Entropy-rigidity estimates}\label{subsec:entropy-rigid}

We next establish the entropy estimates which ensure that MSR almost-iid sources with sublinear defect size have the same asymptotic
spectral entropy as their reference iid state. The point is that,
although the states in the MSR almost-iid class need not be close in trace norm to
\(\rho^{\otimes n}\), their deviation from the tensor-power structure is
controlled by a sublinear defect size. As a consequence, this deviation
has no effect on the entropy rate. The following lemma makes this precise
in the language of information spectrum quantities, and will be used
repeatedly to replace spectral entropy rates of MSR almost-iid sources by the
ordinary von Neumann entropy of the reference state.

{\begin{lemma}[(Spectral entropy rates of MSR almost-iid sources)]
\label{lem:spectral-entropy-msr}
Let \(\rho\in\mathcal D(\mathcal H)\), let \((r_n)_n\) be a
sequence of integers satisfying \(0\leq r_n\leq n\) and
\(r_n=o(n)\), and let \(\bm\rho:=(\rho_n)_n\) satisfy
\bb
\rho_n\in
\mathcal S^n(\mathcal H,\rho^{\otimes(n-r_n)})
\qquad\text{for every }n.
\ee
Then
\bb
\underline S(\bm\rho)=\overline S(\bm\rho)=S(\rho).
\ee
\end{lemma}}
{\begin{proof}
We first prove \(\underline S(\bm\rho)\geq S(\rho)\).
Throughout this proof,
\(H_{\min}^{\varepsilon,\|\cdot\|_1}\) and
\(H_{\min}^{\varepsilon,\mathrm P}\) denote smoothing over
normalized states, with respect to trace norm and purified
distance, respectively.

Theorem~1 of~\cite{Datta_2009}, whose asymptotic identity also
holds for normalized trace-norm smoothing,\footnote{Denote by
\(H_{\min,\leq1}^{\eta,\|\cdot\|_1}(\omega)\) the min-entropy
obtained by trace-norm smoothing over subnormalized states.
For normalized \(\omega\) and \(0<\eta<1/2\), a subnormalized
candidate \(\widetilde\omega\) has trace
\(t\in[1-\eta,1]\).
Normalizing it gives a state within trace norm \(2\eta\)
of \(\omega\), with min-entropy
\(H_{\min}(\widetilde\omega)+\log t\). Hence
\(H_{\min}^{\eta,\|\cdot\|_1}(\omega)
\leq H_{\min,\leq1}^{\eta,\|\cdot\|_1}(\omega)
\leq H_{\min}^{2\eta,\|\cdot\|_1}(\omega)-\log(1-\eta)\).
Applying these inequalities to \(\rho_n\), dividing by \(n\),
taking the limit inferior as \(n\to\infty\), and then letting
\(\eta\downarrow0\) shows that the two conventions give the
same asymptotic rate.}
gives
\bb
\underline S(\bm\rho)
=
\lim_{\varepsilon\downarrow0}
\liminf_{n\to\infty}
\frac1n
H_{\min}^{\varepsilon,\|\cdot\|_1}(\rho_n).
\ee

The finite-size lower bound in the proof of the strong AEP
for MSR almost-iid states, Proposition~4.2
of~\cite{mazzola2026almost}, applied with a trivial conditioning
system, gives, for every fixed \(\varepsilon\in(0,1)\),
\bb\label{eq:msr-uniform-smooth-min-lower}
\frac1n H_{\min}^{\varepsilon,\mathrm P}(\rho_n)
\geq
S(\rho)-\frac{b_n(\varepsilon)}{n},
\ee
where \(b_n(\varepsilon)\geq0\) satisfies
\(b_n(\varepsilon)=o(n)\).
The bound holds uniformly for
\(\rho_n\in\mathcal S^n(\mathcal H,
\rho^{\otimes(n-r_n)})\).
The sequence \(b_n(\varepsilon)\) depends only on
\(\rho\), the dimension, \(\varepsilon\), and the prescribed
defect-size sequence \((r_n)_n\), and not on the particular
state \(\rho_n\).

To justify uniformity, for each fixed \(\alpha>1\), the
finite-size lower bound in the cited proof has the form
\(n^{-1}H_{\min}^{\varepsilon,\mathrm P}(\nu_n)
\geq H_\alpha(\rho)-\delta_{n,\alpha,\varepsilon}\)
for all admissible \(\nu_n\), where
\(H_\alpha(\rho):=(1-\alpha)^{-1}\log\Tr\rho^\alpha\)
and \(\delta_{n,\alpha,\varepsilon}\) is independent of
\(\nu_n\) and tends to zero when \(r_n=o(n)\).
Taking the infimum over the class, then the limit inferior
as \(n\to\infty\), and finally \(\alpha\downarrow1\),
using \(H_\alpha(\rho)\to S(\rho)\), gives the uniform
asymptotic lower bound by \(S(\rho)\).
The positive part of the supremum of
\(nS(\rho)-H_{\min}^{\varepsilon,\mathrm P}(\nu_n)\)
over the class is therefore \(o(n)\), and may be chosen
as \(b_n(\varepsilon)\).
Consequently,
\bb
\liminf_{n\to\infty}
\frac1n H_{\min}^{\varepsilon,\mathrm P}(\rho_n)
\geq S(\rho).
\ee

By~\eqref{eq:h-min}, for \(0<\varepsilon<1/2\),
\bb
H_{\min}^{2\varepsilon,\|\cdot\|_1}(\rho_n)
\geq
H_{\min}^{\varepsilon,\mathrm P}(\rho_n).
\ee
Consequently, for every fixed \(0<\varepsilon<1/2\),
\bb
\liminf_{n\to\infty}
\frac1n H_{\min}^{2\varepsilon,\|\cdot\|_1}(\rho_n)
\geq S(\rho).
\ee
Letting \(\varepsilon\downarrow0\) and using the
information-spectrum identity above therefore yields
\bb
\underline S(\bm\rho)\geq S(\rho).
\ee

We next prove \(\overline S(\bm\rho)\leq S(\rho)\).
Since MSR almost-iid sources with \(r_n=o(n)\) are
weakly almost iid along their reference state,
Lemma~3 of~\cite{datta2026entropy} applies to \(\bm\rho\).
Let \(d:=\dim\mathcal H\), and define, for \(q\in(0,1)\),
\bb
\sigma_q:=(1-q)\rho+\frac{q}{d}I_{\mathcal H},
\qquad
h_q:=-\Tr\rho\log\sigma_q.
\ee
By the cited lemma, \(h_q\to S(\rho)\) as \(q\downarrow0\).
Fix \(R>S(\rho)\), and choose \(q\in(0,1)\) and
\(\delta>0\) such that \(h_q+\delta<R\).
The universal typical projector
\bb
\Pi_n:=
\left\{
-\frac1n\log\sigma_q^{\otimes n}
\leq h_q+\delta
\right\}
\ee
satisfies \(\Tr(\Pi_n\rho_n)\to1\).

Since all logarithms are to base two, the definition of
\(\Pi_n\) gives
\bb
1
\geq \Tr\!\left[\Pi_n\sigma_q^{\otimes n}\right]
\geq 2^{-n(h_q+\delta)}\Tr\Pi_n.
\ee
Hence
\bb
\Tr\Pi_n
\leq 2^{n(h_q+\delta)}
\leq 2^{nR}.
\ee
By the projector characterization of the sup-spectral
entropy rate (see~\eqref{eq:sup-spectral-entropy-rate}
and the remark below it), the existence of these
high-probability projectors implies
\(\overline S(\bm\rho)\leq R\).
Since \(R>S(\rho)\) was arbitrary,
\(\overline S(\bm\rho)\leq S(\rho)\).

Finally, since
\(\underline S(\bm\rho)\leq\overline S(\bm\rho)\),
the two bounds imply
\bb
\underline S(\bm\rho)
=
\overline S(\bm\rho)
=
S(\rho).
\ee
\end{proof}}

\begin{lemma}[(Conditional spectral entropy bound for MSR almost-iid states)]
\label{lem:conditional-spectral-bound-msr}
Let \(\bm\rho_{AB}:=(\rho_{A^nB^n})_n\) be MSR almost-iid along
\(\rho_{AB}\), with defect sizes \(r_n=o(n)\). Then
\bb
\overline S(A|B)_{\bm\rho}
\le
S(A|B)_\rho .
\ee
\end{lemma}

\begin{proof}
By stability of the MSR almost-iid property under taking marginals given by Corollary~\ref{cor:msr-marginals}, the marginal
source \(\bm\rho_B:=(\rho_{B^n})_n\) is MSR almost-iid along
\(\rho_B\), with defect sizes \(r_n=o(n)\). Hence,
Lemma~\ref{lem:spectral-entropy-msr}, applied to \(\bm\rho_{AB}\) and to
\(\bm\rho_B\), gives
\bb
\overline S(\bm\rho_{AB})=S(\rho_{AB}),
\qquad
\underline S(\bm\rho_B)=S(\rho_B).
\ee
Using the information-spectrum chain-rule inequality given by Proposition~9 of~\cite{bowen2006beyond},
\bb
\overline S(A|B)_{\bm\rho}
\le
\overline S(\bm\rho_{AB})-\underline S(\bm\rho_B),
\ee
we obtain
\bb
\overline S(A|B)_{\bm\rho}
\le
S(\rho_{AB})-S(\rho_B)
=
S(A|B)_\rho .
\ee
\end{proof}

The entropy-rigidity estimates above show that MSR almost-iid sources
retain the relevant asymptotic entropy behaviour of their reference state.
In the applications below, however, we need a uniform version of this
estimate: the weight assigned to the spectral projection
\(\{\rho_n\ge 2^{-n\gamma} I_n\}\) must vanish uniformly over the whole
MSR almost-iid class with a prescribed sublinear defect bound, whenever
\(\gamma<S(\rho)\). Equivalently, no source in this class can
asymptotically assign non-negligible weight to eigenvalues larger than
\(2^{-n\gamma}\) for any \(\gamma<S(\rho)\). This is the form in which the
entropy-rigidity estimate will be used later, and is made precise in the following corollary.
\begin{cor}[(Uniform lower spectral tail bound for MSR almost-iid sources)]
\label{cor:msr-lower-spectral-tail}
Let \(\rho\in\mathcal D(\mathcal H)\), and let \((r_n)_n\)
be a sequence of integers satisfying \(0\leq r_n\leq n\)
and \(r_n=o(n)\). Then, for every \(\gamma<S(\rho)\),
\bb
\sup_{\omega_n\in
\mathcal S^n(\mathcal H,\rho^{\otimes(n-r_n)})}
\Tr\!\left[
\left\{\omega_n\geq 2^{-n\gamma}I_n\right\}\omega_n
\right]
\longrightarrow0
\qquad\text{as }n\to\infty.
\ee
Here \(I_n\) is the identity operator on
\(\mathcal H^{\otimes n}\).
Since the MSR classes are nested in the defect parameter\footnote{By Definition~\ref{def:MSR}, the MSR almost-iid classes satisfy
\(
\mathcal S^n(\mathcal H,\rho^{\otimes(n-s)})
\subseteq
\mathcal S^n(\mathcal H,\rho^{\otimes(n-r)})
\forall \, 0\leq s\leq r\leq n.
\)
Indeed, every vector in the set \(\mathcal V\) defining a defect
space with parameter \(s\) also belongs to the corresponding
set with parameter \(r\), since the additional \(r-s\)
unrestricted tensor positions may carry the reference purification.  The same MSR extension therefore
remains admissible.},
the supremum includes all MSR almost-iid states along \(\rho\)
with defect size at most \(r_n\).
\end{cor}

\begin{proof}
Fix \(\gamma<S(\rho)\), choose
\(\gamma<\kappa<S(\rho)\), and let \(\varepsilon\in(0,1)\).
The uniform lower bound
\eqref{eq:msr-uniform-smooth-min-lower}, established in the
proof of Lemma~\ref{lem:spectral-entropy-msr}, gives
\bb
H_{\min}^{\varepsilon,\mathrm P}(\omega_n)
\geq nS(\rho)-b_n(\varepsilon),
\ee
where \(b_n(\varepsilon)\geq0\) is \(o(n)\) and is
independent of the particular admissible state \(\omega_n\).
Consequently, for all sufficiently large \(n\),
\bb
H_{\min}^{\varepsilon,\mathrm P}(\omega_n)\geq n\kappa
\ee
uniformly over the class.

The normalized smoothing ball is compact and the
min-entropy is continuous on normalized states, so the
supremum in its definition is attained. Thus, for each
admissible \(\omega_n\), there exists a normalized state
\(\widetilde\omega_n\) such that
\bb
P(\widetilde\omega_n,\omega_n)\leq\varepsilon,
\qquad
\widetilde\omega_n\leq 2^{-n\kappa}I_n.
\ee
The Fuchs--van de Graaf inequalities give
\(\|\widetilde\omega_n-\omega_n\|_1\leq2\varepsilon\).

Set \(Q_n:=\{\omega_n\geq2^{-n\gamma}I_n\}\).
Since \(\Tr\omega_n=1\),
\(\operatorname{rank}Q_n\leq2^{n\gamma}\), and hence
\bb
\Tr[Q_n\widetilde\omega_n]
\leq
2^{-n\kappa}\operatorname{rank}Q_n
\leq
2^{-n(\kappa-\gamma)}.
\ee
Using \(0\leq Q_n\leq I_n\), we obtain
\bb
\begin{aligned}
\Tr[Q_n\omega_n]
&\leq
\Tr[Q_n\widetilde\omega_n]
+\|\omega_n-\widetilde\omega_n\|_1\\
&\leq
2^{-n(\kappa-\gamma)}+2\varepsilon.
\end{aligned}
\ee
Taking the supremum over the class and then the limit
superior yields
\bb
\limsup_{n\to\infty}
\sup_{\omega_n\in
\mathcal S^n(\mathcal H,\rho^{\otimes(n-r_n)})}
\Tr[Q_n\omega_n]
\leq2\varepsilon.
\ee
Letting \(\varepsilon\downarrow0\) proves the claim.
\end{proof}

{We shall also need to control the effect of auxiliary systems on
conditional spectral entropy rates. Such systems arise in our
auxiliary description of MSR almost-iid sources. The following
lemma shows that adding a quantum register to Alice's system and
classical labels to the conditioning system cannot increase the
sup-conditional spectral entropy rate, provided that the quantum
register has subexponential dimension in the blocklength.
}
\begin{lemma}[(Finite-register stability of conditional spectral entropy)]
\label{lem:finite-register-stability-conditional-spectrum}
Let
\(\bm\gamma_{ABCZ}:=(\gamma_{A_lB_lC_lZ_l})_{l\ge1}\)
be a sequence of states, where \(Z_l\) is a classical register, as in the applications below, and let
\(d_{C,l}:=\dim\mathcal H_{C_l}\). Then
\bb
\overline S(AC|BZ)_{\bm\gamma}
\le
\overline S(A|B)_{\bm\gamma}
+
\limsup_{l\to\infty}\frac1l\log d_{C,l}.
\ee
In particular, if \(\log d_{C,l}=o(l)\), then
\(\overline S(AC|BZ)_{\bm\gamma}\le
\overline S(A|B)_{\bm\gamma}\). Thus adjoining to Alice's side a system
whose logarithmic dimension is sublinear, and adjoining a classical
register to the conditioning system, does not increase the normalized
sup-conditional spectral entropy rate.
\end{lemma}
\begin{proof}
    See Appendix~\ref{app:aux-info-spec}.
\end{proof}

\section{Entanglement manipulation: distillable entanglement and entanglement cost}
\label{sec:ent-manipulate}

In this section we recall the basic entanglement manipulation tasks that
will be considered in the sequel. We first discuss entanglement
distillation, and its pure-state specialization, entanglement concentration,
and then turn to the reverse task of entanglement dilution. We introduce
the corresponding operational quantities for general sequences of
bipartite states and recall their standard expressions in the iid
setting.

{Throughout this paper, an LOCC map means a deterministic completely
positive trace-preserving map implementable by local operations (LO) and
classical communication (CC), with classical communication allowed in both
directions.

Let
\(
\bm\rho_{AB}:=(\rho_{A^nB^n})_{n\ge1}
\)
be a sequence of bipartite states shared by two distant parties, Alice and
Bob. An entanglement-distillation protocol for this sequence consists, for
each \(n\), of an LOCC map \(\pazocal D_n\) acting on \(A^nB^n\), whose
aim is to extract a maximally entangled state \(\Phi_{M_n}\) of Schmidt
rank \(M_n\). A rate \(R\ge0\) is said to be achievable if there exists
such a sequence of LOCC maps \((\pazocal D_n)_n\) such that
\bb
\liminf_{n\to\infty}\frac1n\log M_n\ge R
\ee
and
\bb
\lim_{n\to\infty}
F\!\left(
\pazocal D_n(\rho_{A^nB^n}),
\Phi_{M_n}
\right)
=
1.
\ee
The distillable entanglement of the sequence \(\bm\rho_{AB}\) is defined
as
\bb
E_D(\bm\rho_{AB})
:=
\sup\{R:\ R\text{ is achievable for }\bm\rho_{AB}\}.
\ee
When the sequence is iid, \(\bm\rho_{AB}=(\rho_{AB}^{\otimes n})_n\),
this reduces to the usual distillable entanglement \(E_D(\rho_{AB})\).
For general mixed states, no single-letter formula for \(E_D(\rho_{AB})\)
is known in general, although the hashing inequality gives the coherent
information as an achievable lower bound~\cite{DevetakWinter2005}. Thus, in particular,
\bb
E_D(\rho_{AB})\ge \max\{0,I(A\rangle B)_\rho\},
\ee
where $I(A\rangle B)_\rho:= - S(A|B)_\rho$ denotes the coherent information. For pure states, however, the
problem is completely characterized: if
\(
\psi_{AB}:=\ket{\psi}\!\bra{\psi}_{AB}
\),
then
\bb
E_D(\psi_{AB})
=
S(\rho_A),
\qquad
\rho_A:=\Tr_B\psi_{AB}.
\ee

In the special case where each \(\rho_{A^nB^n}\) is pure, say
\(
\rho_{A^nB^n}
=
\Psi_{A^nB^n}
:=
\ket{\Psi_n}\!\bra{\Psi_n}_{A^nB^n}
\),
the same task is usually called \emph{entanglement concentration}. Thus,
for a pure-state sequence
\(
\bm\Psi_{AB}:=(\ket{\Psi_n}\!\bra{\Psi_n}_{A^nB^n})_n
\),
a rate \(R\ge0\) is achievable for entanglement concentration if there
exist LOCC maps \(\pazocal C_n\) such that
\bb
\liminf_{n\to\infty}\frac1n\log M_n\ge R,
\qquad
\lim_{n\to\infty}
F\!\left(
\pazocal C_n(\Psi_{A^nB^n}),
\Phi_{M_n}
\right)
=
1.
\ee
The optimal rate is denoted by \(E_D(\bm\Psi_{AB})\), or simply referred
to as the concentration rate of the sequence. In the iid pure-state case,
\(\bm\Psi_{AB}=(\psi_{AB}^{\otimes n})_n\), this rate is equal to
the entropy of entanglement \(S(\rho_A)\), by the entanglement
concentration theorem of Bennett, Bernstein, Popescu and
Schumacher~\cite{BennettBernsteinPopescuSchumacher1996}.
For arbitrary pure-state sequences, entanglement concentration
admits an information-spectrum
characterization~\cite{Hayashi2006}.

Conversely, entanglement dilution is the task of creating a desired
bipartite state from maximally entangled states by LOCC. Let
\(
\bm\rho_{AB}:=(\rho_{A^nB^n})_{n\ge1}
\)
be a sequence of bipartite states. An entanglement-dilution protocol for
this sequence consists, for each \(n\), of an LOCC map \(\pazocal E_n\)
whose input is a maximally entangled state \(\Phi_{M_n}\) of Schmidt rank
\(M_n\), and whose output is a state on \(A^nB^n\). A rate \(R\ge0\) is
said to be achievable if there exists such a sequence of LOCC maps
\((\pazocal E_n)_n\) such that
\bb
\limsup_{n\to\infty}\frac1n\log M_n\le R
\ee
and
\bb
\lim_{n\to\infty}
F\!\left(
\pazocal E_n(\Phi_{M_n}),
\rho_{A^nB^n}
\right)
=
1.
\ee
The \emph{entanglement cost} of the sequence \(\bm\rho_{AB}\) is defined
as
\bb
E_C(\bm\rho_{AB})
:=
\inf\{R:\ R\text{ is achievable for entanglement dilution of }
\bm\rho_{AB}\}.
\ee

{In the iid case,
\(\bm\rho_{AB}=(\rho_{AB}^{\otimes n})_{n\geq1}\),
this recovers the usual entanglement cost \(E_C(\rho_{AB})\).
Its value is given by the regularized entanglement of
formation~\cite{HaydenHorodeckiTerhal2001},
\bb
E_C(\rho_{AB})
=
E_F^\infty(\rho_{AB})
:=
\lim_{m\to\infty}
\frac1m E_F(\rho_{AB}^{\otimes m})
=
\inf_{m\geq1}
\frac1m E_F(\rho_{AB}^{\otimes m}).
\ee
The limit exists and equals the infimum by subadditivity of
\(E_F(\rho_{AB}^{\otimes m})\) in \(m\).}
Here the entanglement of formation of a bipartite state \(\omega_{AB}\)
is defined by
\bb
E_F(\omega_{AB})
:=
\inf_{\{p_i,\ket{\psi_i}\}}
\sum_i p_i\, S(\psi_{i,A}),
\ee
where the infimum is over all finite pure-state ensemble decompositions
\bb
\omega_{AB}
=
\sum_i p_i
\ket{\psi_i}\!\bra{\psi_i}_{AB},
\ee
and
\(
\psi_{i,A}:=\Tr_B\ket{\psi_i}\!\bra{\psi_i}_{AB}
\).
Thus \(E_F^\infty(\rho_{AB})\) is the asymptotic rate at which maximally
entangled states must be consumed, per copy, in order to prepare
\(\rho_{AB}^{\otimes n}\) by LOCC.

\section{Universal entanglement concentration for pure MSR almost-iid sources}
\label{sec:ent-conc}

Fix a bipartite pure state
\(\ket{\phi}_{AB}\), let
\(\phi_{AB}:=\ket{\phi}\!\bra{\phi}_{AB}\), and set
\(\phi_A:=\Tr_B\phi_{AB}\).
{Let \((r_n)_n\) be a fixed sequence of integers satisfying
\(0\leq r_n\leq n\) and \(r_n=o(n)\), and denote by} 
\(
\mathfrak S^{\mathrm{pMSR}}_{\phi}((r_n)_n)
\)
the class of all pure-state sequences
$
\bm\Psi_{AB}
:=
\bigl(\ketbra{\Psi_{A^nB^n}}\bigr)_n
$,
which are MSR almost-iid along \(\phi_{AB}\), with defect size at
most \(r_n\) at blocklength \(n\). By~\eqref{eq:pure-pi-msr}, every such
sequence satisfies
\(
\ket{\Psi_{A^nB^n}}
\in
\Sym^n(\mathcal H_A\otimes\mathcal H_B)
\)
for each \(n\).

A sequence of entanglement-concentration protocols is universal over
\(\mathfrak S^{\mathrm{pMSR}}_{\phi}((r_n)_n)\) if it is independent of
the particular source sequence in this class. We shall prove that every
rate below \(S(\phi_A)\) is achievable uniformly over the class. In
fact, the protocol constructed below depends only on the desired rate
and the local Hilbert spaces; the reference state enters only through
the admissible-rate condition \(R<S(\phi_A)\).

We begin with the following consequence of Schur--Weyl duality.

\begin{lemma}[Lemma~C.1 of~\cite{Fawzi-Renner}]
\label{lem:Fawzi-Renner-Schur}
Let \(\mathcal H_A\) and \(\mathcal H_B\) be Hilbert spaces with
\(\dim\mathcal H_A=\dim\mathcal H_B=d\), and let
\(\Lambda_{n,d}\) denote the set of Young diagrams of size \(n\) with at
most \(d\) rows. For \(X\in\{A,B\}\), write the Schur--Weyl
decomposition as
\bb
\mathcal H_X^{\otimes n}
\simeq
\bigoplus_{\lambda\in\Lambda_{n,d}}
\mathcal U_{X,\lambda}\otimes\mathcal V_{X,\lambda}.
\ee
Then there exists a family
\bb
\left\{
\ket{\psi_\lambda}_{
\mathcal V_{A,\lambda}\mathcal V_{B,\lambda}}
\right\}_{\lambda\in\Lambda_{n,d}}
\ee
of normalized maximally entangled vectors such that every vector
\(\ket{\Omega}\in\Sym^n(\mathcal H_A\otimes\mathcal H_B)\) admits a
decomposition
\bb
\ket{\Omega}
=
\sum_{\lambda\in\Lambda_{n,d}}
\ket{\omega_\lambda}_{
\mathcal U_{A,\lambda}\mathcal U_{B,\lambda}}
\otimes
\ket{\psi_\lambda}_{
\mathcal V_{A,\lambda}\mathcal V_{B,\lambda}},
\ee
where the vectors
\(\ket{\omega_\lambda}_{
\mathcal U_{A,\lambda}\mathcal U_{B,\lambda}}\)
need not be normalized.
\end{lemma}

\begin{rem}
The Schur--Weyl spaces in
Lemma~\ref{lem:Fawzi-Renner-Schur} depend on the blocklength \(n\).
This dependence, as well as the corresponding dependence of their
dimensions and of the vectors \(\ket{\psi_\lambda}\), will be
suppressed in the notation.
\end{rem}

\begin{rem}\label{rem:equal-local-dimensions}
If \(d_A:=\dim\mathcal H_A\) and \(d_B:=\dim\mathcal H_B\) are unequal,
we embed the smaller Hilbert space isometrically into a Hilbert space of
dimension \(d:=\max\{d_A,d_B\}\). Such local isometric embeddings do not
change the entanglement or the achievable concentration rates and allow
Lemma~\ref{lem:Fawzi-Renner-Schur} to be applied.
{These tensor-power local isometries also preserve the pure MSR
almost-iid structure and its defect bound. The entropy of entanglement
of the reference state is unchanged. Consequently,
Lemma~\ref{lem:schur-block-concentration-pure-msr} can be applied
on the embedded local spaces with the same defect bound and
admissible-rate condition.}
\end{rem}

Lemma~\ref{lem:Fawzi-Renner-Schur} is the basis of the universal
concentration protocol of Matsumoto and Hayashi
\cite{PhysRevA.75.062338}. After the local Schur transforms and the
measurement of the Young-diagram label \(\lambda\), the parties obtain
perfectly correlated labels, and the conditional state on
\(\mathcal V_{A,\lambda}\otimes\mathcal V_{B,\lambda}\) is maximally
entangled with Schmidt rank \(\dim\mathcal V_{A,\lambda}\). It therefore
suffices to show that the probability of obtaining a label for which
this dimension is too small vanishes uniformly over the MSR almost-iid class.

\begin{lemma}[(Uniform Schur-block lower-tail bound)]
\label{lem:schur-block-concentration-pure-msr}
Let \(\ket{\phi}_{AB}\in\mathcal H_A\otimes\mathcal H_B\) be a
bipartite pure state, let
\(\phi_A:=\Tr_B\ket{\phi}\!\bra{\phi}_{AB}\), and and let \((r_n)_n\) be a sequence of integers satisfying
\(0\leq r_n\leq n\) and \(r_n=o(n)\). Consider the Schur--Weyl
decomposition
\bb
\mathcal H_A^{\otimes n}
=
\bigoplus_{\lambda\in\Lambda_{n,d_A}}
\mathcal U_{A,\lambda}\otimes\mathcal V_{A,\lambda},
\ee
and let \(P_{A,\lambda}\) denote the projector onto
\(\mathcal U_{A,\lambda}\otimes\mathcal V_{A,\lambda}\).

For each
\bb
\bm\Psi
=
\bigl(\ketbra{\Psi_{A^nB^n}}\bigr)_n
\in
\mathfrak S^{\mathrm{pMSR}}_{\phi}((r_n)_n),
\ee
set
\bb
\rho_{A^n}^{\Psi}
:=
\Tr_{B^n}\ket{\Psi_{A^nB^n}}\!\bra{\Psi_{A^nB^n}}.
\ee
Then, for every \(0\leq R<S(\phi_A)\),
\bb
\sup_{\bm\Psi\in
\mathfrak S^{\mathrm{pMSR}}_{\phi}((r_n)_n)}
\sum_{\lambda:\,\frac1n\log\dim\mathcal V_{A,\lambda}<R}
\Tr\!\left[
P_{A,\lambda}\rho_{A^n}^{\Psi}
\right]
\longrightarrow 0
\qquad\text{as }n\to\infty.
\ee
\end{lemma}

\begin{proof}
By stability of the MSR almost-iid property under local tensor-power channels,
Lemma~\ref{lem:msr-stability-local-channels}, applied to the partial
trace channel
\(
\Tr_B:
\mathcal D(\mathcal H_A\otimes\mathcal H_B)
\longrightarrow
\mathcal D(\mathcal H_A),
\)
the reduced sequence
\(\bm\rho_A^{\Psi}:=(\rho_{A^n}^{\Psi})_n\) is MSR almost-iid
along \(\phi_A\), with defect size at most \(r_n\).

Consequently, Corollary~\ref{cor:msr-lower-spectral-tail} implies that,
for every \(\gamma<S(\phi_A)\),
\bb\label{eq:unif-conv}
\sup_{\bm\Psi\in
\mathfrak S^{\mathrm{pMSR}}_{\phi}((r_n)_n)}
\Tr\!\left[
\left\{
\rho_{A^n}^{\Psi}
\geq
2^{-n\gamma}I_{A^n}
\right\}
\rho_{A^n}^{\Psi}
\right]
\longrightarrow 0
\qquad\text{as }n\to\infty.
\ee

Fix \(R<S(\phi_A)\), and choose \(R'\) and \(\gamma\) such that
\(
R<R'<\gamma<S(\phi_A).
\)
Define
\bb
\Pi_n^{<R}
:=
\sum_{\lambda:\,\frac1n\log\dim\mathcal V_{A,\lambda}<R}
P_{A,\lambda}.
\ee

Since \(d_A\) is fixed, the number of Young diagrams
\(\lvert\Lambda_{n,d_A}\rvert\) is polynomial in \(n\). Moreover, the
Weyl dimension formula gives a polynomial upper bound on
\(\dim\mathcal U_{A,\lambda}\), uniformly in
\(\lambda\in\Lambda_{n,d_A}\). Hence there is a polynomial \(p\) such
that
\bb
\sum_{\lambda\in\Lambda_{n,d_A}}
\dim\mathcal U_{A,\lambda}
\leq p(n).
\ee
It follows that
\begin{align}
\Tr\Pi_n^{<R}
& =
\sum_{\lambda:\,\frac1n\log\dim\mathcal V_{A,\lambda}<R}
\dim\mathcal U_{A,\lambda}\,
\dim\mathcal V_{A,\lambda} 
\leq
p(n)2^{nR}.
\end{align}
Since \(R'>R\) and \(p(n)\) is polynomial, for all sufficiently
large \(n\),
\bb
\Tr\Pi_n^{<R}
\leq
p(n)2^{nR}
\leq
2^{nR'}.
\ee

For each source \(\bm\Psi\) in the prescribed class, define
\bb
Q_n^{\Psi}
:=
\left\{
\rho_{A^n}^{\Psi}
\geq
2^{-n\gamma}I_{A^n}
\right\}.
\ee
Since \(Q_n^{\Psi}\) is a spectral projector of
\(\rho_{A^n}^{\Psi}\), it commutes with that state. Hence
\bb
\begin{aligned}
\rho_{A^n}^{\Psi}
&=
Q_n^{\Psi}\rho_{A^n}^{\Psi}Q_n^{\Psi}
+
(I_{A^n}-Q_n^{\Psi})
\rho_{A^n}^{\Psi}
(I_{A^n}-Q_n^{\Psi})
\\
&\leq
Q_n^{\Psi}\rho_{A^n}^{\Psi}Q_n^{\Psi}
+
2^{-n\gamma}I_{A^n}.
\end{aligned}
\ee
The inequality follows because all eigenvalues on the
complementary spectral subspace are smaller than
\(2^{-n\gamma}\).

Taking the trace against the positive operator
\(\Pi_n^{<R}\), we obtain
\bb
\begin{aligned}
\Tr\!\left[
\Pi_n^{<R}\rho_{A^n}^{\Psi}
\right]
&\leq
\Tr\!\left[
\Pi_n^{<R}
Q_n^{\Psi}\rho_{A^n}^{\Psi}Q_n^{\Psi}
\right]
+
2^{-n\gamma}\Tr\Pi_n^{<R}
\\
&\leq
\Tr\!\left[
Q_n^{\Psi}\rho_{A^n}^{\Psi}
\right]
+
2^{-n(\gamma-R')}.
\end{aligned}
\ee
Here we used \(0\leq\Pi_n^{<R}\leq I_{A^n}\), the positivity
of \(Q_n^{\Psi}\rho_{A^n}^{\Psi}Q_n^{\Psi}\), and the rank
bound above. In particular, no commutation between
\(\Pi_n^{<R}\) and \(\rho_{A^n}^{\Psi}\) is required.

The rank bound and its onset are independent of the particular
source. Therefore, taking the supremum over the prescribed class
gives, for all sufficiently large \(n\),
\bb
\begin{aligned}
&\sup_{\bm\Psi\in
\mathfrak S^{\mathrm{pMSR}}_{\phi}((r_n)_n)}
\Tr\!\left[
\Pi_n^{<R}\rho_{A^n}^{\Psi}
\right]
\\
&\qquad\leq
\sup_{\bm\Psi\in
\mathfrak S^{\mathrm{pMSR}}_{\phi}((r_n)_n)}
\Tr\!\left[
Q_n^{\Psi}\rho_{A^n}^{\Psi}
\right]
+
2^{-n(\gamma-R')}
\longrightarrow 0.
\end{aligned}
\ee
The first term tends to zero by \eqref{eq:unif-conv}, and
the second tends to zero because \(\gamma>R'\).

Finally, by the definition of \(\Pi_n^{<R}\),
\bb
\Tr\!\left[
\Pi_n^{<R}\rho_{A^n}^{\Psi}
\right]
=
\sum_{\lambda:\,\frac1n\log\dim\mathcal V_{A,\lambda}<R}
\Tr\!\left[
P_{A,\lambda}\rho_{A^n}^{\Psi}
\right].
\ee
This proves the asserted uniform convergence.
\end{proof}

\begin{thm}[(Universal entanglement concentration for pure MSR almost-iid sources)]
\label{thm:universal-concentration-pi-pure-msr}
Let \(\ket{\phi}_{AB}\in\mathcal H_A\otimes\mathcal H_B\) be a bipartite
pure state, let
\(
\rho_A:=\Tr_B\ket{\phi}\!\bra{\phi}_{AB}
\),
and let \((r_n)_n\) be a sequence of integers satisfying
\(0\leq r_n\leq n\) and \(r_n=o(n)\). Let
\(\mathfrak S^{\rm pMSR}_{\phi}((r_n)_n)\) denote the class of pure MSR almost-iid
sources along \(\ket{\phi}_{AB}\) with defect size at most \(r_n\) at
blocklength \(n\).

Then every rate
\(
0\le R<S(\rho_A)
\)
is universally achievable for entanglement concentration over
\(\mathfrak S^{\rm pMSR}_{\phi}((r_n)_n)\). More precisely, for every such
\(R\), there exists a sequence of LOCC protocols \((\mathcal C_n)_n\),
depending only on \(R\) and the local Hilbert spaces, and not on the
particular source in the class, such that, with
\(M_n:=\lfloor 2^{nR}\rfloor\),
\bb
\sup_{\bm\Psi=(\Psi_n)_n \in\mathfrak S^{\rm pMSR}_{\phi}((r_n)_n)}
\left[
1-
F\!\left(
\mathcal C_n(\Psi_n),
\Phi_{M_n}
\right)
\right]
\to 0
\quad\text{as } n\to\infty .
\ee
The protocol itself depends only on the target rate and the local Hilbert spaces;
the reference state enters only through the admissible range \(R<S(\rho_A)\).
\end{thm}

\begin{proof}
Fix \(0\le R<S(\rho_A)\), and set
\(
M_n:=\lfloor 2^{nR}\rfloor
\).
The protocol is the Schur--Weyl concentration protocol of
Hayashi and Matsumoto~\cite{hayashi2002universal}.
{If \(d_A\neq d_B\), we first apply the fixed local isometric
embeddings described in Remark~\ref{rem:equal-local-dimensions},
so that both local spaces have dimension
\(d:=\max\{d_A,d_B\}\).}
In what follows, we identify
the source and reference states with their embedded images.
Alice and Bob then perform the local Schur measurements
\(
\{P_{A,\lambda}\}_{\lambda}
\)
and
\(
\{P_{B,\lambda}\}_{\lambda}
\),
where \(P_{A,\lambda}\) and \(P_{B,\lambda}\) are the
orthogonal projections onto
\(\mathcal U_{A,\lambda}\otimes\mathcal V_{A,\lambda}\)
and
\(\mathcal U_{B,\lambda}\otimes\mathcal V_{B,\lambda}\),
respectively.
{Let
\(
\bm\Psi=(\Psi_n)_{n\geq1}
\in
\mathfrak S^{\mathrm{pMSR}}_{\phi}((r_n)_n)
\),
where
\(
\Psi_n:=\ket{\Psi_n}\!\bra{\Psi_n}_{A^nB^n}.
\)}
Since \(\ket{\Psi_n}
\in
\Sym^n\!\left(
\mathcal H_A\otimes\mathcal H_B,
\ket{\phi}_{AB}^{\otimes(n-r_n)}
\right)\), and hence 
is permutation-invariant,
Lemma~\ref{lem:Fawzi-Renner-Schur} gives
\bb
\ket{\Psi_n}
=
\sum_{\lambda}
\ket{\varphi_{n,\lambda}}_{\mathcal U_{A,\lambda}\mathcal U_{B,\lambda}}
\otimes
\ket{\psi_{\lambda}}_{\mathcal V_{A,\lambda}\mathcal V_{B,\lambda}},
\ee
where each \(\ket{\psi_{\lambda}}\) is a normalized maximally entangled
state on
\(
\mathcal V_{A,\lambda}\otimes\mathcal V_{B,\lambda}
\),
and the vectors \(\ket{\varphi_{n,\lambda}}\) are not necessarily
normalized.

Thus only matching Young diagrams occur on Alice's and Bob's sides. In
particular, the two Schur measurements produce the same label
\(\lambda\) with probability one. Conditional on outcome \(\lambda\),
Alice and Bob discard the \(\mathcal U\)-systems and retain an exact
maximally entangled state on
\(
\mathcal V_{A,\lambda}\otimes\mathcal V_{B,\lambda}
\),
whose Schmidt rank is
\(
\dim\mathcal V_{A,\lambda}=\dim\mathcal V_{B,\lambda}
\).

{Alice and Bob compare their measurement outcomes by classical
communication. When the labels agree and
\(\dim\mathcal V_{A,\lambda}\geq M_n\), they apply a fixed
deterministic LOCC map taking
\(\ket{\psi_\lambda}\!\bra{\psi_\lambda}\) to
\(\Phi_{M_n}\).} Such a channel exists by Nielsen's
majorization theorem~\cite{Nielsen1999}.
In all other cases, they discard their quantum systems and
prepare a fixed product state \(\sigma_n\) on the output
spaces of \(\Phi_{M_n}\). Since \(\ket{\psi_\lambda}\)
is fixed by the Schur--Weyl decomposition, the conditional
conversion channels can be chosen independently of the source
and the reference state. The entire protocol depends only on
\(R\) and the local Hilbert spaces.

Let
\(
\rho_{A^n}:=\Tr_{B^n}\ket{\Psi_n}\!\bra{\Psi_n}
\).
Since only matching Young diagrams occur, we have
{\bb
(P_{A,\lambda}\otimes I_{B^n})\ket{\Psi_n}
=
(I_{A^n}\otimes P_{B,\lambda})\ket{\Psi_n}
=
(P_{A,\lambda}\otimes P_{B,\lambda})\ket{\Psi_n}.
\ee}
Hence, the probability of obtaining the common Schur label \(\lambda\) is
\(
\Tr[P_{A,\lambda}\rho_{A^n}]
\).
Therefore the failure probability is
\bb
p_{\rm fail}^{(n)}(\Psi_n)
=
\sum_{\lambda:\,\dim\mathcal V_{A,\lambda}<M_n}
\Tr\!\left[
P_{A,\lambda}\rho_{A^n}
\right].
\ee
Since \(M_n=\lfloor 2^{nR}\rfloor\), the condition
\(
\dim\mathcal V_{A,\lambda}<M_n
\)
implies
\(
\frac1n\log\dim\mathcal V_{A,\lambda}<R
\).
Therefore, by Lemma~\ref{lem:schur-block-concentration-pure-msr},
\bb
\sup_{\bm\Psi\in\mathfrak S^{\rm pMSR}_{\phi}((r_n)_n)}
p_{\rm fail}^{(n)}(\Psi_n)
\to 0
\quad\text{as } n\to\infty .
\ee

For each input source, the output of the protocol has the form
\bb
(1-p_{\rm fail}^{(n)}(\Psi_n))\Phi_{M_n}
+
p_{\rm fail}^{(n)}(\Psi_n)\sigma_n,
\ee
{where \(\sigma_n\) is the fixed product state specified above.} 

Since
\(\Phi_{M_n}\) is pure, for any state \(\omega\),
\(
F(\omega,\Phi_{M_n})=\sqrt{\Tr[\Phi_{M_n}\omega]}
\).
Thus, for \(p=p_{\rm fail}^{(n)}(\Psi_n)\),
\bb
F\!\left(
(1-p)\Phi_{M_n}+p\sigma_n,\Phi_{M_n}
\right)
=
\sqrt{(1-p)+p\,\Tr[\Phi_{M_n}\sigma_n]}
\ge
\sqrt{1-p}.
\ee
It follows that
\bb
F\!\left(
\mathcal C_n(\Psi_n),
\Phi_{M_n}
\right)
\ge
\sqrt{1-p_{\rm fail}^{(n)}(\Psi_n)}.
\ee
Since \(p_{\rm fail}^{(n)}(\Psi_n)\to0\) uniformly over
\(\mathfrak S^{\rm pMSR}_{\phi}((r_n)_n)\), the fidelity converges to
\(1\) uniformly over this class.

Finally, since \(M_n=\lfloor 2^{nR}\rfloor\), we have
\bb
\liminf_{n\to\infty}\frac1n\log M_n=R.
\ee
Thus the rate \(R\) is universally achievable over
\(\mathfrak S^{\rm pMSR}_{\phi}((r_n)_n)\).
\end{proof}

%%% START INSERT
\section{Universal entanglement distillation for MSR almost-iid sources}
\label{sec:MSR-universal-distillation}

We now consider entanglement distillation from arbitrary, possibly
mixed, MSR almost-iid sources. We show that every rate below the two-way
distillable entanglement of the reference state is achievable by a
single LOCC protocol sequence, uniformly over each prescribed
sublinear defect size.

{Fix \(\rho_{AB}\in\mathcal D(\mathcal H_A\otimes\mathcal H_B)\).
For a prescribed sequence of integers \((r_n)_n\) satisfying
\(0\leq r_n\leq n\) and \(r_n=o(n)\), let
\(\mathfrak S_{\rho}^{\mathrm{MSR}}((r_n)_n)\) denote the class
of all sources
\(\boldsymbol{\omega}_{AB}=(\omega_{A^nB^n})_{n\geq1}\)
satisfying}
\bb
\omega_{A^nB^n}\in
\mathcal S^n\!\left(
\mathcal H_A\otimes\mathcal H_B,
\rho_{AB}^{\otimes(n-r_n)}
\right)
\qquad\text{for every }n.
\label{eq:distillation-MSR-class}
\ee
Universality means that the protocol is independent of the
particular source in this class. The protocol constructed below
may depend on the reference $\rho_{AB}$ and the target rate.

{The proof combines an operator upper bound for MSR almost-iid
input states with an error-transfer lemma for a pure target. Our Lemma~\ref{lem:distillation-error-transfer} shows that,
if two input states, $\omega$ and $\sigma$, satisfy \(\omega\le C\sigma\), the error
measured by \(1-F^2\) for a fixed protocol with a pure target
is at most \(C\) times as large on input \(\omega\) as on
input \(\sigma\).
The remaining ingredient is a single iid distillation protocol
whose error decreases exponentially in the blocklength, uniformly
on \(\sigma_{AB}^{\otimes n}\) for all bipartite states
\(\sigma_{AB}\) in a fixed trace-distance neighbourhood of
\(\rho_{AB}\). To obtain this exponential estimate, we follow
the local robustness construction in the proof of Lemma~S6
in the Supplementary Material of Tirone et al.~\cite{Tirone26}
and retain the explicit finite-block hashing estimate from
their Lemma~S5, Eqs.~(S29)--(S30).}

\begin{lemma}[(Operator domination of MSR almost-iid states)]
\label{lem:distillation-MSR-domination}
{Let \(\mathcal H\) have finite dimension \(d\), let
\(\rho\in\mathcal D(\mathcal H)\), and let
\(n\in\mathbb N_+\) and \(r\in\{0,\ldots,n\}\).}
Suppose that $\omega_n\in\mathcal D(\mathcal H^{\otimes n})$
has a normalized extension $\kappa_{H^nE^n}$ and a unit-vector
purification $|\theta\rangle_{HE}$ of $\rho$ such that
\bb
\omega_n=(\id_{H^n}\otimes\operatorname{Tr}_{E^n})(\kappa),
\qquad
\operatorname{supp}\kappa\subseteq\mathcal M_n(\theta,r),
\ee
where $\mathcal M_n(\theta,r)$ is the defect space defined earlier.
The finite dimension of $E$ is arbitrary.

For $0<q<1$, define
\bb
\tau_q:=(1-q)\rho+q\frac{I_H}{d},
\qquad
L_{n,r}:=\sum_{j=0}^r\binom nj,
\qquad
C_{n,r}(q):=
L_{n,r}(1-q)^{-n}\left(\frac dq\right)^r.
\ee
Then
\bb
\omega_n\leq C_{n,r}(q)\tau_q^{\otimes n}.
\label{eq:distillation-MSR-domination}
\ee
In particular, this holds for every MSR almost-iid state along $\rho$
with defect size at most $r$.

For every fixed \(q\in(0,1)\) and prescribed sequence of integers
\((r_n)_n\) satisfying \(0\leq r_n\leq n\) and \(r_n=o(n)\), put
\bb
\gamma(q):=-\log(1-q),
\qquad
a_n:=\log L_{n,r_n}+r_n\log(d/q).
\label{eq:distillation-domination-exponents}
\ee
Then $a_n=o(n)$ and
\bb
\omega_n\leq2^{n\gamma(q)+a_n}\tau_q^{\otimes n},
\label{eq:distillation-MSR-asymptotic-domination}
\ee
uniformly over the class with defect size at most $r_n$.
\end{lemma}

\begin{proof}
Define
\bb
P:=|\theta\rangle\langle\theta|,
\qquad Q:=I_{HE}-P,
\ee
and, for $S\subseteq[n]$, let
\bb
P_S:=\bigotimes_{i=1}^n R_i(S),
\qquad
R_i(S):=
\begin{cases}
Q,&i\in S,\\
P,&i\notin S.
\end{cases}
\ee
These projections are pairwise orthogonal, and
\bb
\mathcal M_n(\theta,r)
=\bigoplus_{\substack{S\subseteq[n]\\|S|\leq r}}
\operatorname{Ran}P_S.
\ee
Indeed, fixing a set $T$ of $r$ unrestricted positions and
expanding $I=P+Q$ on those positions yields the sectors indexed
by $S\subseteq T$. Conversely, every $S$ with $|S|\leq r$ is
contained in such a $T$.

Let $\Pi_{n,r}:=\sum_{|S|\leq r}P_S$. The support assumption gives
$\kappa=\Pi_{n,r}\kappa\Pi_{n,r}$. For every vector $|v\rangle$,
the Cauchy--Schwarz inequality yields
\bb
\begin{aligned}
\langle v|\kappa|v\rangle
&=\left\|\sum_{|S|\leq r}\kappa^{1/2}P_S|v\rangle\right\|^2\\
&\leq L_{n,r}\sum_{|S|\leq r}
\|\kappa^{1/2}P_S|v\rangle\|^2.
\end{aligned}
\ee
Consequently,
\bb
\kappa\leq L_{n,r}\sum_{|S|\leq r}P_S\kappa P_S.
\label{eq:distillation-defect-pinching}
\ee
{The sum on the right is obtained by removing the off-diagonal
blocks of \(\kappa\) between the subspaces
\(\operatorname{Ran}P_S\).

Define \(p_S:=\operatorname{Tr}(P_S\kappa P_S)\).
Then \(p_S\geq0\) and \(\sum_{|S|\leq r}p_S=1\).
If \(p_S>0\), the normalized state
\(p_S^{-1}P_S\kappa P_S\) is supported on
\(\operatorname{Ran}P_S\). It therefore has the fixed pure factor
\(|\theta\rangle\langle\theta|\) at every position outside
\(S\), and hence
\bb
P_S\kappa P_S
=
p_S\left(
\bigotimes_{i\notin S}
|\theta\rangle\langle\theta|_{H_iE_i}
\right)\otimes\eta_S
\ee
for some density operator \(\eta_S\) on \((HE)^S\).
If \(p_S=0\), the corresponding positive operator is zero.
The tensor product in this expression is understood in the
original site order: \(\eta_S\) occupies the positions in \(S\),
and the copies of \(|\theta\rangle\langle\theta|\) occupy
the positions in \(S^c\).}

Applying $\id_{H^n}\otimes\operatorname{Tr}_{E^n}$ to
\eqref{eq:distillation-defect-pinching}, and writing
$\sigma_S:=\operatorname{Tr}_{E^S}\eta_S$, gives
\bb
\omega_n\leq
L_{n,r}\sum_{\substack{|S|\leq r\\p_S>0}}
p_S\,\rho^{\otimes S^c}\otimes\sigma_S.
\ee
Since $\sigma_S$ is a density operator,
$\sigma_S\leq I_{H^S}$.
The positive tensor-product expansion of $\tau_q^{\otimes n}$
contains, for $s=|S|$, the term
\bb
(1-q)^{n-s}\left(\frac qd\right)^s
\rho^{\otimes S^c}\otimes I_{H^S}.
\ee
It follows that
\bb
\begin{aligned}
\rho^{\otimes S^c}\otimes\sigma_S
&\leq
(1-q)^{-(n-s)}
\left(\frac dq\right)^s\tau_q^{\otimes n}\\
&\leq
(1-q)^{-n}
\left(\frac dq\right)^r\tau_q^{\otimes n}.
\end{aligned}
\ee
The last step uses $s\leq r$, $0<1-q<1$, and $d/q>1$.
Summing the weights $p_S$ proves
\eqref{eq:distillation-MSR-domination}.

For the asymptotic assertion, the binomial-volume bound gives
$L_{n,r}\leq2^{nh(r/n)}$ for $r\leq n/2$.
For completeness, when $x=r/n\in(0,1/2]$,
\bb
1\geq
\sum_{j=0}^r\binom nj x^j(1-x)^{n-j}
\geq L_{n,r}x^r(1-x)^{n-r},
\ee
which proves the bound; for $r=0$, $L_{n,0}=1$.
Since $r_n=o(n)$, eventually $r_n\leq n/2$ and
\bb
0\leq\frac{a_n}{n}
\leq h(r_n/n)
+\frac{r_n}{n}\log(d/q)
\longrightarrow0.
\ee
Thus $C_{n,r_n}(q)=2^{n\gamma(q)+a_n}$ with $a_n=o(n)$.
These quantities depend only on $d,q$ and the defect size,
so the estimate is uniform. Permutation invariance of $\kappa$
was not needed.
\end{proof}

\begin{lemma}[(Finite error transfer)]
\label{lem:distillation-error-transfer}
Let $\omega,\sigma$ be normalized states satisfying
$\omega\leq C\sigma$, and let $\mathcal D$ be a completely
positive trace-preserving map with a pure target $\Phi$.
Define
\bb
e_\sigma:=1-F(\mathcal D(\sigma),\Phi)^2,
\qquad
u:=\min\{1,Ce_\sigma\}.
\ee
Then
\bb
\begin{aligned}
1-F(\mathcal D(\omega),\Phi)^2
&\leq u,\\
1-F(\mathcal D(\omega),\Phi)
&\leq1-\sqrt{1-u}\leq u,\\
T(\mathcal D(\omega),\Phi)
&\leq P(\mathcal D(\omega),\Phi)\leq\sqrt u.
\end{aligned}
\label{eq:distillation-finite-error-transfer}
\ee
\end{lemma}

{\begin{proof}
The adjoint \(\mathcal D^*\) is positive and unital.
Let \(I_{\mathrm{in}}\) and \(I_{\mathrm{out}}\) denote the
identity operators on the input and output spaces of
\(\mathcal D\), respectively. Define the operator
\bb
Q:=\mathcal D^*(I_{\mathrm{out}}-\Phi)
=I_{\mathrm{in}}-\mathcal D^*(\Phi).
\ee
Since \(0\leq I_{\mathrm{out}}-\Phi\leq I_{\mathrm{out}}\),
positivity of \(\mathcal D^*\) and
\(\mathcal D^*(I_{\mathrm{out}})=I_{\mathrm{in}}\) imply
\bb
0\leq Q\leq I_{\mathrm{in}}.
\ee
Purity of the target and trace preservation give, for every
normalized input \(\xi\),
\bb
1-F(\mathcal D(\xi),\Phi)^2
=\operatorname{Tr}(Q\xi).
\ee
Therefore
\bb
\operatorname{Tr}(Q\omega)
\leq C\operatorname{Tr}(Q\sigma)
=Ce_\sigma.
\ee
Indeed,
\(\operatorname{Tr}[Q^{1/2}(C\sigma-\omega)Q^{1/2}]\geq0\).
Squared infidelity is also at most one, proving the first bound.
The fidelity bound follows by taking square roots, and
the last bound follows from the Fuchs--van de Graaf inequality.
\end{proof}}

For a pure target, squared infidelity is bounded by half trace
distance:
\bb
1-F(\zeta,\Phi)^2
=\operatorname{Tr}[(I-\Phi)(\zeta-\Phi)]
\leq T(\zeta,\Phi).
\label{eq:distillation-pure-target-test}
\ee
Thus an exponential iid half-trace-distance estimate supplies
the same exponent and prefactor for squared infidelity.
If an iid estimate is instead stated for $1-F$,
the conversion
$1-F^2=(1-F)(1+F)\leq2(1-F)$
changes only the prefactor.

\begin{lemma}[Uniform exponential iid distillation]
\label{lem:distillation-local-iid}
Let $\rho$ be a finite-dimensional bipartite state and let
$0<R<E_D(\rho)$. Set $M_n:=\lfloor2^{nR}\rfloor$.
There exist constants $\eta,K,\beta>0$, an integer $n_0$,
and a single sequence of deterministic LOCC maps
$\mathcal D_n$ such that
\bb
\sup_{\substack{
\sigma\in\mathcal D(\mathcal H_A\otimes\mathcal H_B)\\
T(\sigma,\rho)\leq\eta}}
T\bigl(
\mathcal D_n(\sigma^{\otimes n}),\Phi_{M_n}
\bigr)
\leq K2^{-n\beta}
\qquad(n\geq n_0).
\label{eq:distillation-local-iid-exponential}
\ee
The same bound holds for $1-F^2$ and $1-F$.
The channels and constants may depend on $\rho,R$ and the local
Hilbert spaces, but not on $\sigma$.
\end{lemma}

\begin{proof}
Choose $R<R_0<R_1<E_D(\rho)$ and fix $\delta\in(0,1/4)$
such that
\bb
(1-4\delta)R_1>R_0.
\label{eq:distillation-block-accuracy}
\ee
By the operational definition of $E_D(\rho)$, we may choose
arbitrarily large $k$ and an LOCC map $\Gamma$ with
\bb
\ell:=\lceil R_1k\rceil,
\qquad
\operatorname{Tr}\!\left[
\Phi_2^{\otimes\ell}\Gamma(\rho^{\otimes k})
\right]\geq1-\delta.
\label{eq:distillation-reference-block}
\ee
One may first use an achievable rate strictly above $R_1$
and convert its maximally entangled target to $\ell$ ebits.

Put $g_\ell(x):=h(x)+x\log(4^\ell-1)$.
We can choose the same $k$ sufficiently large that
\bb
\begin{aligned}
\ell-g_\ell(2\delta)
&\geq(1-4\delta)\ell-h(2\delta)\\
&\geq(1-4\delta)R_1k-h(2\delta)\\
&>R_0k.
\end{aligned}
\label{eq:distillation-isotropic-rate-gap}
\ee
Here $\log(4^\ell-1)\leq2\ell$, and the strict final
inequality follows for all sufficiently large $k$ from
\eqref{eq:distillation-block-accuracy}.
Fix $k,\Gamma$ and $\ell$ from now on.

Set $\eta:=\delta/k$.
For $T(\sigma,\rho)\leq\eta$, tensor-product telescoping
and contractivity give
\bb
T\bigl(
\Gamma(\sigma^{\otimes k}),\Gamma(\rho^{\otimes k})
\bigr)
\leq k\,T(\sigma,\rho)\leq\delta.
\ee
The target projection is a measurement effect, so
\bb
f(\sigma):=
\operatorname{Tr}\!\left[
\Phi_2^{\otimes\ell}\Gamma(\sigma^{\otimes k})
\right]
\geq1-2\delta=:f_0.
\label{eq:distillation-block-shell}
\ee

Let $\mathcal T_\ell$ be isotropic twirling on the output
of $\Gamma$. It is LOCC and preserves overlap with
$\Phi_2^{\otimes\ell}$, hence
\bb
(\mathcal T_\ell\circ\Gamma)(\sigma^{\otimes k})
=\omega_{f(\sigma)},
\qquad
\omega_f:=
f\Phi_2^{\otimes\ell}
+(1-f)\frac{I-\Phi_2^{\otimes\ell}}{4^\ell-1}.
\ee
Thus one fixed channel sends the entire neighbourhood into
the shell $f\in[f_0,1]$.

By continuity and the strict gap in
\eqref{eq:distillation-isotropic-rate-gap}, choose a fixed
$\alpha>0$ such that
\bb
2\delta+\alpha\leq1-4^{-\ell},
\qquad
v:=\ell-g_\ell(2\delta+\alpha)-\alpha>R_0k.
\label{eq:distillation-hashing-slack}
\ee
These choices are possible because both inequalities have
strict slack at $\alpha=0$.
Also, $f_0>1/2>4^{-\ell}$.

{Apply the finite-block estimate in Lemma~S5 of~\cite{Tirone26},
Eqs.~(S29)--(S30), with \(q_0=f_0\), \(t=\alpha\), and
\(r=2\delta+\alpha\) in the notation of that paper. For} 
\bb
c_N:=
\left\lceil
N[g_\ell(2\delta+\alpha)+\alpha]
\right\rceil,
\qquad
m_N:=N\ell-c_N,
\ee
{it supplies an LOCC map \(\mathcal G_N\),
independent of \(f\), whenever \(m_N>0\), such that}
\bb
\sup_{f\in[f_0,1]}
T\bigl(
{\mathcal G}_N(\omega_f^{\otimes N}),
\Phi_2^{\otimes m_N}
\bigr)
\leq\exp(-2N\alpha^2)+2^{-N\alpha}.
\label{eq:distillation-finite-hashing}
\ee
The shared randomness is averaged over and every measurement
branch is retained, so this is a trace-preserving protocol.

For $n$ input copies, set $N_n:=\lfloor n/k\rfloor$ and
discard the unused copies.
Applying $\mathcal T_\ell\circ\Gamma$ independently to the
$N_n$ blocks produces exactly
$\omega_{f(\sigma)}^{\otimes N_n}$ for an iid input.
Moreover,
\bb
m_{N_n}\geq N_nv-1
\geq\frac vk n-v-1.
\ee
Since $v/k>R_0>R$, choose an integer
\bb
n_0\geq
\max\left\{
k,\,
\left\lceil\frac{v+2}{v/k-R}\right\rceil
\right\}.
\ee
Then $m_{N_n}\geq\lceil Rn\rceil$ for every $n\geq n_0$.
Consequently, $2^{m_{N_n}}\geq M_n$, and there is a
deterministic LOCC map $\mathcal P_n$ mapping
$\Phi_2^{\otimes m_{N_n}}$ exactly to $\Phi_{M_n}$.
This is the standard deterministic conversion from a
maximally entangled state to one of smaller Schmidt rank.

Let $\mathcal Q_n$ discard the unused input copies and define
\bb
\mathcal D_n:=
\mathcal P_n\circ{\mathcal G}_{N_n}
\circ(\mathcal T_\ell\circ\Gamma)^{\otimes N_n}
\circ\mathcal Q_n
\qquad(n\geq n_0).
\ee
For $n<n_0$, let $\mathcal D_n$ output a fixed product state
in the required output space.
Contractivity and \eqref{eq:distillation-finite-hashing} give
\bb
\sup_{T(\sigma,\rho)\leq\eta}
T\bigl(
\mathcal D_n(\sigma^{\otimes n}),\Phi_{M_n}
\bigr)
\leq\exp(-2N_n\alpha^2)+2^{-N_n\alpha}
\qquad(n\geq n_0).
\ee
Set
\bb
b:=\min\left\{\frac{2\alpha^2}{\ln2},\alpha\right\}>0,
\qquad
\beta:=\frac bk>0,
\qquad
K:=2^{1+b}.
\label{eq:distillation-local-iid-constants}
\ee
Since $N_n\geq n/k-1$,
\bb
\exp(-2N_n\alpha^2)+2^{-N_n\alpha}
\leq2\,2^{-bN_n}
\leq K2^{-n\beta}.
\ee
This proves \eqref{eq:distillation-local-iid-exponential}.
Equation \eqref{eq:distillation-pure-target-test} gives the
same bound for squared infidelity, and $1-F\leq1-F^2$
gives the infidelity bound.
\end{proof}

All parameters in Lemma~\ref{lem:distillation-local-iid}
are fixed before the depolarizing parameter $q$ is chosen.
In natural-exponential notation its exponent is
$\beta\ln2=k^{-1}\min\{2\alpha^2,\alpha\ln2\}$.

\begin{thm}[(Universal entanglement distillation for MSR almost-iid sources)]
\label{thm:MSR-universal-distillation}
Let $\rho_{AB}\in\mathcal D(\mathcal H_A\otimes\mathcal H_B)$
be a finite-dimensional bipartite state, and let
$0\leq r_n\leq n$ be a prescribed integer sequence with
$r_n=o(n)$.
Then every rate $0\leq R<E_D(\rho_{AB})$ is universally
achievable for entanglement distillation over
$\mathfrak S_{\rho}^{\mathrm{MSR}}((r_n)_n)$.

More precisely, for each such $R$, there is a single sequence
of deterministic LOCC maps $(\mathcal D_n)_n$,
depending only on $\rho_{AB},R$ and the local Hilbert spaces,
such that, with $M_n:=\lfloor2^{nR}\rfloor$,
\bb
\sup_{\boldsymbol{\omega}_{AB}\in
\mathfrak S_{\rho}^{\mathrm{MSR}}((r_n)_n)}
\left[
1-F\bigl(
\mathcal D_n(\omega_{A^nB^n}),\Phi_{M_n}
\bigr)
\right]\longrightarrow0.
\label{eq:distillation-universal-MSR-fidelity}
\ee
The channels can be chosen independently of $(r_n)_n$;
uniform convergence is asserted separately for each
prescribed sublinear defect size.

In particular, every source in this class satisfies
\bb
E_D(\boldsymbol{\omega}_{AB})\geq E_D(\rho_{AB}).
\label{eq:distillation-MSR-rate-robustness}
\ee
\end{thm}

\begin{proof}
For $R=0$, output $\Phi_1$.
Suppose $0<R<E_D(\rho_{AB})$.
Lemma~\ref{lem:distillation-local-iid} supplies one sequence
$\mathcal D_n$, constants $\eta,K,\beta>0$, and an integer
$n_0$, with uniform iid squared-infidelity bound
$K2^{-n\beta}$.

Choose a fixed
\bb
0<q<\min\{\eta,1-2^{-\beta}\},
\qquad
\tau_q:=(1-q)\rho_{AB}+q\frac{I_{AB}}d,
\qquad
d:=\dim(\mathcal H_A\otimes\mathcal H_B).
\ee
Then $T(\tau_q,\rho_{AB})\leq q<\eta$ and
\bb
\Delta:=\beta-\gamma(q)
=\beta+\log(1-q)>0.
\ee
Thus the iid bound applies to $\tau_q^{\otimes n}$.
Lemmas~\ref{lem:distillation-MSR-domination} and
\ref{lem:distillation-error-transfer} imply, for
$n\geq n_0$,
\bb
\sup_{\boldsymbol{\omega}_{AB}\in
\mathfrak S_{\rho}^{\mathrm{MSR}}((r_n)_n)}
\left[
1-F\bigl(
\mathcal D_n(\omega_{A^nB^n}),\Phi_{M_n}
\bigr)^2
\right]
\leq u_n,
\qquad
u_n:=\min\{1,K2^{-n\Delta+a_n}\},
\label{eq:distillation-uniform-MSR-error}
\ee
where
$a_n=\log L_{n,r_n}+r_n\log(d/q)=o(n)$.
Hence $u_n\to0$, and $1-F\leq1-F^2$ proves
\eqref{eq:distillation-universal-MSR-fidelity}.

All protocol choices, including $q$, are independent of
the defect size and of the particular source.
Finally,
\bb
\lim_{n\to\infty}\frac1n\log M_n=R.
\ee
Every rate below $E_D(\rho_{AB})$ is therefore achievable
for every source in the class.
Taking the supremum over these rates proves
\eqref{eq:distillation-MSR-rate-robustness} when
$E_D(\rho_{AB})>0$.
When $E_D(\rho_{AB})=0$, the inequality follows from
nonnegativity of distillable entanglement.
\end{proof}

For $0<R<E_D(\rho_{AB})$, the proof also gives quantitative
uniform bounds:
\bb
\begin{aligned}
\sup_{\boldsymbol{\omega}_{AB}\in
\mathfrak S_{\rho}^{\mathrm{MSR}}((r_n)_n)}
[1-F(\mathcal D_n(\omega_{A^nB^n}),\Phi_{M_n})]
&\leq1-\sqrt{1-u_n}\leq u_n,\\
\sup_{\boldsymbol{\omega}_{AB}\in
\mathfrak S_{\rho}^{\mathrm{MSR}}((r_n)_n)}
T(\mathcal D_n(\omega_{A^nB^n}),\Phi_{M_n})
&\leq\sqrt{u_n}.
\end{aligned}
\ee
The same last bound holds for purified distance.
For every $\varepsilon\in(0,\Delta)$, eventually
$a_n\leq n\varepsilon$.
Thus $1-F$ and $1-F^2$ are bounded by
$K2^{-n(\Delta-\varepsilon)}$, and half trace distance
and purified distance by
$\sqrt K\,2^{-n(\Delta-\varepsilon)/2}$.
The threshold may depend on the prescribed defect sequence.

There is also an exact class-level rate statement.
Let $E_D^{\mathrm{univ}}(\mathfrak S)$ denote the supremum
of rates achievable by one LOCC protocol sequence with
error vanishing uniformly over the source class
$\mathfrak S$.
For each prescribed $r_n=o(n)$,
\bb
\begin{aligned}
E_D^{\mathrm{univ}}\!\left(
\mathfrak S_{\rho}^{\mathrm{MSR}}((r_n)_n)
\right)
&=
\inf_{\boldsymbol{\omega}_{AB}\in
\mathfrak S_{\rho}^{\mathrm{MSR}}((r_n)_n)}
E_D(\boldsymbol{\omega}_{AB})\\
&=E_D(\rho_{AB}).
\end{aligned}
\label{eq:distillation-optimal-universal-rate}
\ee
Indeed, Theorem~\ref{thm:MSR-universal-distillation}
gives the lower bounds, while the class contains the exact
iid source $(\rho_{AB}^{\otimes n})_n.$
Any universally achievable rate must be achievable for
that source, whose distillable entanglement is
$E_D(\rho_{AB})$.

%%% END INSERT
{\section{Robustness of entanglement dilution for MSR almost-iid states}
\label{sec:msr-ent-dil}

We now turn to the reverse task, entanglement dilution. While the proof
of the entanglement concentration theorem
(Theorem~\ref{thm:universal-concentration-pi-pure-msr}) relies on
Schur--Weyl duality and the Hayashi--Matsumoto protocol, entanglement
dilution is naturally approached through information-spectrum techniques,
together with the MSR almost-iid entropy-rigidity results established earlier.

We begin with a lemma which provides the technical bridge between the MSR almost-iid
structure of the target states \(\rho_{A^nB^n}\) and the ensemble-based
description needed for entanglement dilution. In the iid setting, a
cq-extension of \(\rho_{AB}\) with pure conditional states,
\(
\rho_{RAB}
=
\sum_i p_i \ket{i}\!\bra{i}_R
\otimes
\ket{\psi_i}\!\bra{\psi_i}_{AB},
\)
induces, for each \(n\), the tensor-product ensemble
\(
\rho_{AB}^{\otimes n}
=
\sum_{i^n} p_{i^n}
\ket{\psi_{i^n}}\!\bra{\psi_{i^n}}_{A^nB^n},
\)
where
\(
p_{i^n}:=p_{i_1}\cdots p_{i_n}
\)
and
\(
\ket{\psi_{i^n}}
:=
\ket{\psi_{i_1}}\otimes\cdots\otimes\ket{\psi_{i_n}}.
\)
Equivalently, this gives a natural cq-extension of
\(\rho_{AB}^{\otimes n}\), with a classical register \(R^n\) recording
the sequence \(i^n\) of ensemble labels.

For an MSR almost-iid target sequence \((\rho_{A^nB^n})_n\), however,
the states \(\rho_{A^nB^n}\) need not arise from such an exact
tensor-product ensemble. The lemma below shows that this difficulty can
be overcome: any fixed cq-extension of the reference state \(\rho_{AB}\)
with pure conditional states can be lifted to cq-extensions of the MSR almost-iid
states \(\rho_{A^nB^n}\), at the cost of an additional classical register
\(\widetilde R_n\) whose logarithmic dimension is \(o(n)\). This overhead
is asymptotically negligible at the level of rates, and the resulting
marginal sequence on \(R^nA^nB^n\) is MSR almost-iid along the chosen
cq-extension of \(\rho_{AB}\). This allows the spectral-entropy rigidity
results proved earlier to be applied to the lifted cq-state, which is the
form required in the subsequent entanglement-dilution argument. 

\begin{lemma}[(cq-lifting of MSR almost-iid target sequences up to sublinear classical overhead)]
\label{lem:cq-lifting}
Let
\(
\bm\rho_{AB}:=(\rho_{A^nB^n})_{n}
\)
be an MSR almost-iid target sequence along \(\rho_{AB}\), with integer
defect sizes \(0\leq r_n\leq n\) satisfying \(r_n=o(n)\). Let
\bb
\varrho_{RAB}
=
\sum_i p_i \ket{i}\!\bra{i}_R
\otimes
\ket{\phi_i}\!\bra{\phi_i}_{AB}
\ee
be a finite-dimensional cq-extension of \(\rho_{AB}\) with pure conditional states.
Then there exists a sequence of cq-extensions $(\widehat\varrho^{(n)})_n$ where 
{\bb
\widehat\varrho^{(n)} \equiv \widehat\varrho_{\widetilde R_nR^nA^nB^n}
=
\sum_{\substack{\alpha\in\mathcal A_n,\; i^n:\\
q^{(n)}_{\alpha,i^n}>0}}
q^{(n)}_{\alpha,i^n}
\ket{\alpha}\!\bra{\alpha}_{\widetilde R_n}
\otimes
\ket{i^n}\!\bra{i^n}_{R^n}
\otimes
\ket{\phi^{(n)}_{\alpha,i^n}}\!
\bra{\phi^{(n)}_{\alpha,i^n}}_{A^nB^n}
\ee}
of \(\rho_{A^nB^n}\) such that
\(
\log\dim\widetilde R_n=o(n),
\)
and, writing
\bb
\varrho_{R^nA^nB^n}
:=
\Tr_{\widetilde R_n}
\widehat\varrho_{\widetilde R_nR^nA^nB^n},
\ee
the sequence
\(
\bm\varrho_{RAB}
:=
(\varrho_{R^nA^nB^n})_{n}
\)
is MSR almost-iid along \(\varrho_{RAB}\), with defect sizes
\(r_n=o(n)\).
\end{lemma}
\begin{proof}
See Appendix~\ref{app:cq}.
\end{proof}

The following lemma is a simple but useful blocking observation. Fix a
block size \(m\), and write each blocklength as
\(
n=lm+s
\),
where \(0\le s<m\). We group the first \(lm\) tensor positions into
\(l\) blocks of size \(m\), leaving \(s\) remaining tensor positions.
Since \(m\) is fixed, this remainder has size uniformly bounded in
\(n\). Each unrestricted tensor position in one of the vectors spanning
the MSR almost-iid support space can affect at most one \(m\)-block. Hence, after
blocking, such a vector has at most \(\min\{r_n,l\}\) unrestricted
\(m\)-blocks, while the final \(s\) tensor positions are treated as an
unrestricted remainder. Thus, for each fixed value of \(s\), the MSR almost-iid
structure is preserved after blocking, with block-defect size
\(t_{l,s}=o(l)\). This will allow us to apply MSR almost-iid entropy estimates to
fixed-size blocks while keeping the leftover tensor positions
asymptotically negligible.

{\begin{lemma}[(Blocking with bounded remainder)]
\label{lem:MSR-blocking-remainder}
Let
\(
\bm\rho_{AB}:=(\rho_{A^nB^n})_{n\ge1}
\)
be MSR almost-iid along \(\rho_{AB}\), with integer defect sizes
\(0\leq r_n\leq n\) satisfying \(r_n=o(n)\).
Fix \(m\in\mathbb N_+\) and a purification
\(\ket{\theta}_{ABE}\) of \(\rho_{AB}\), and set
\(\ket{\theta_m}_{A^mB^mE^m}:=\ket{\theta}_{ABE}^{\otimes m}\).
For \(n\geq m\), write \(n=lm+s\), where
\(l\geq1\) and \(0\leq s<m\). Then there exists an extension
\bb
\omega_{(A^mB^mE^m)^l A^sB^sE^s}
\ee
of \(\rho_{A^{lm+s}B^{lm+s}}\) such that
\bb
\supp\omega_{(A^mB^mE^m)^l A^sB^sE^s}
\subseteq
\operatorname{span}
\mathcal V\!\left(
(\mathcal H_{ABE}^{\otimes m})^{\otimes l},
\ket{\theta_m}^{\otimes(l-t_{l,s})}
\right)
\otimes
\mathcal H_{ABE}^{\otimes s},
\ee
where
\(
t_{l,s}:=\min\{r_{lm+s},l\}
\).
In particular, for each fixed residue \(s\), we have
\(t_{l,s}=o(l)\) as \(l\to\infty\). Moreover, \(\omega\) is invariant
under permutations of the \(l\) blocks of size \(m\).
\end{lemma}
\begin{proof}
See Appendix~\ref{app:msr-block}.
\end{proof}
}

The following corollary combines the preceding blocking lemma with the
cq-lifting construction. Fixing a block size \(m\), we view the source at
blocklength \(n=lm+s\) as consisting of \(l\) blocks of size \(m\), together
with a remainder of size \(s<m\). The blocking lemma shows that, for each
fixed \(s\), the first \(lm\) systems form an MSR almost-iid source at the block level,
with only \(o(l)\) unrestricted blocks, while the remaining \(s\) systems
have fixed size and are therefore asymptotically negligible. We then apply the cq-lifting argument to these \(m\)-blocks, relative to
a chosen cq-extension of \(\rho_{AB}^{\otimes m}\) with pure conditional states. The result
is a cq-extension of the full state
\(\rho_{A^{lm+s}B^{lm+s}}\) with pure conditional states and an additional classical
register \(\widetilde R_l\), in the sense that
\(\log\dim\widetilde R_l=o(l)\). After tracing out this additional
register and the final \(A^sB^s\)-systems, the marginal on
\(R^l(A^mB^m)^l\) remains MSR almost-iid along the chosen block-level
cq-extension.

\begin{cor}[(cq-lifting with bounded remainder)]
\label{cor:cq-lifting-remainder}
Fix \(m\in\mathbb N_+\) and \(s\in\{0,\ldots,m-1\}\). Let
\((\rho_{A^{lm+s}B^{lm+s}})_{l\ge1}\) be as in
Lemma~\ref{lem:MSR-blocking-remainder}. Let
\bb
\varrho^{(m)}_{RA^mB^m}
=
\sum_i p_i^{(m)}
\ket{i}\!\bra{i}_R
\otimes
\ket{\phi_i^{(m)}}\!\bra{\phi_i^{(m)}}_{A^mB^m}
\ee
be a finite-dimensional cq-extension of \(\rho_{AB}^{\otimes m}\) with pure conditional states. Then there
exist cq-extensions with pure conditional states
\bb
\widehat\varrho^{(l,s)}_{\widetilde R_lR^lA^{lm+s}B^{lm+s}}
\ee
of \(\rho_{A^{lm+s}B^{lm+s}}\) such that
\bb
\log\dim\widetilde R_l=o(l),
\ee
and, after tracing out \(\widetilde R_l\) and the final
\(A^sB^s\)-systems, the marginal sequence on
\(R^l(A^mB^m)^l\) is MSR almost-iid along
\(\varrho^{(m)}_{RA^mB^m}\), with defect sizes \(o(l)\).
\end{cor}
\begin{proof}
See Appendix~\ref{app:lift-rem}.
\end{proof}

{

\begin{thm}[(Robust achievability of entanglement dilution for MSR almost-iid sources)]
\label{thm:MSR-dilution}
Let \(\rho_{AB}\in\mathcal D(\mathcal H_A\otimes\mathcal H_B)\), and let
\(\bm\rho_{AB}:=(\rho_{A^nB^n})_{n\ge1}\) be an MSR almost-iid target
sequence along \(\rho_{AB}\), with integer defect sizes
\(0\leq r_n\leq n\) satisfying \(r_n=o(n)\). Then
\bb
E_C(\bm\rho_{AB})
\le
E_F^\infty(\rho_{AB}).
\ee
Equivalently, every rate \(R>E_F^\infty(\rho_{AB})\) is achievable for
entanglement dilution of \(\bm\rho_{AB}\). Moreover, for every fixed
defect-size sequence \(r_n=o(n)\),
\bb
\sup_{\bm\rho_{AB}}
E_C(\bm\rho_{AB})
\le
E_F^\infty(\rho_{AB}),
\ee
where the supremum is over all MSR almost-iid target sequences along
\(\rho_{AB}\) with defect size at most \(r_n\) at blocklength \(n\).
\end{thm}

\begin{proof}
We use the information-spectrum characterization of entanglement cost
given by Theorem~1 of~\cite{bowen2011entanglement}. For any sequence
\(\bm\omega_{AB}:=(\omega_{A^nB^n})_{n\ge1}\),
\bb
E_C(\bm\omega_{AB})
=
\inf_{\bm\gamma\in D_{\rm cq}(\bm\omega_{AB})}
\overline S(A|R)_{\bm\gamma},
\ee
where the infimum is over all cq-extension sequences with pure conditional states for
\(\bm\omega_{AB}\), and \(\overline S(A|R)\) denotes the spectral
sup-conditional entropy rate.

Thus it suffices to construct a sequence of cq-extensions
with pure conditional states for \(\bm\rho_{AB}\), and to bound
its spectral sup-conditional entropy.

In the blocked argument below, spectral rates denoted by
\(\overline S_{\rm block}\) are normalized by the number \(l\) of
\(m\)-blocks; the rates entering \(E_C\) are normalized by the original
blocklength \(n=lm+s\).

Fix \(m\in\mathbb N_+\) and \(\eta>0\). Choose a pure-state ensemble
decomposition
\bb
\rho_{AB}^{\otimes m}
=
\sum_i p_i^{(m)}
\ket{\phi_i^{(m)}}\!\bra{\phi_i^{(m)}}_{A^mB^m}
\ee
such that
\bb
\sum_i p_i^{(m)}
S(\phi^{(m)}_{i,A^m})
\le
E_F(\rho_{AB}^{\otimes m})+\eta .
\ee

Let
\bb
\varrho^{(m)}_{RA^mB^m}
=
\sum_i p_i^{(m)}
\ket{i}\!\bra{i}_R
\otimes
\ket{\phi_i^{(m)}}\!\bra{\phi_i^{(m)}}_{A^mB^m}.
\ee

Then
\bb
S(A^m|R)_{\varrho^{(m)}}
=
\sum_i p_i^{(m)}
S(\phi^{(m)}_{i,A^m})
\le
E_F(\rho_{AB}^{\otimes m})+\eta .
\ee

Write \(n=lm+s\), where \(0\le s<m\). We first fix a value
\(s\in\{0,\ldots,m-1\}\) of the remainder in this decomposition.

By Corollary~\ref{cor:cq-lifting-remainder}, applied to the block-level
cq-extension \(\varrho^{(m)}_{RA^mB^m}\), there exist cq-extensions with pure conditional states
\bb
\widehat\varrho^{(l,s)}_{\widetilde R_lR^lA^{lm+s}B^{lm+s}}
\ee
of \(\rho_{A^{lm+s}B^{lm+s}}\) such that
\bb
\log\dim\widetilde R_l=o(l).
\ee

Moreover, after tracing out \(\widetilde R_l\) and the final
\(A^sB^s\)-systems, the marginal sequence on \(R^l(A^mB^m)^l\) is MSR
almost-iid along \(\varrho^{(m)}_{RA^mB^m}\), with defect sizes
\(o(l)\).

Before applying the conditional spectral entropy bound, we trace out the
\(B^{ml}\)-systems from this blocked marginal sequence. By stability of
the MSR almost-iid condition under local channels (Lemma~\ref{lem:msr-stability-local-channels}), in particular, under partial trace, the resulting sequence on
\(R^l(A^m)^l\) is MSR almost-iid along the bipartite reference state
\bb
\varrho^{(m)}_{RA^m}
=
\Tr_{B^m}\varrho^{(m)}_{RA^mB^m},
\ee
with defect sizes \(o(l)\). Let this blocked marginal sequence be denoted
by
\bb
\bm\sigma^{(s)}
:=
\bigl(\sigma^{(l,s)}_{R^l(A^m)^l}\bigr)_{l\ge1}.
\ee

By the conditional spectral entropy bound for MSR almost-iid sources,
Lemma~\ref{lem:conditional-spectral-bound-msr}, applied to the bipartite
blocked marginal sequence \(\bm\sigma^{(s)}\), we have
\bb
\overline S_{\rm block}(A^m|R)_{\bm\sigma^{(s)}}
\le
S(A^m|R)_{\varrho^{(m)}} ,
\ee
where \(\overline S_{\rm block}\) denotes the spectral sup-conditional
entropy rate normalized by the number \(l\) of \(m\)-blocks.

We now compare \(\bm\sigma^{(s)}\) with the full cq-extension sequence
for the fixed value of \(s\),
\bb
\bm{\widehat\varrho}^{(s)}
:=
\bigl(
\widehat\varrho^{(l,s)}_{\widetilde R_lR^lA^{lm+s}B^{lm+s}}
\bigr)_{l\ge1}.
\ee
Passing from \(\bm\sigma^{(s)}\) to
\(\bm{\widehat\varrho}^{(s)}\) adds the classical register
\(\widetilde R_l\) to the conditioning system and adds the final
\(A^s\)-system to Alice's side. Since
\bb
\log\dim\widetilde R_l=o(l)
\qquad\text{and}\qquad
\log\dim\mathcal H_A^{\otimes s}=s\log d_A=O(1)
\ee
for fixed \(s<m\), Lemma~\ref{lem:finite-register-stability-conditional-spectrum}
implies
\bb
\overline S_{\rm block}
(A^mA^s|\widetilde R R)_{\bm{\widehat\varrho}^{(s)}}
\le
\overline S_{\rm block}(A^m|R)_{\bm\sigma^{(s)}} .
\ee
Here and below, the conditional spectral entropy rates of
\(\bm{\widehat\varrho}^{(s)}\) are evaluated on the marginal
obtained by tracing out Bob's systems.

Combining the two inequalities, we obtain, for each fixed value of \(s\),
\bb
\overline S_{\rm block}
(A^mA^s|\widetilde R R)_{\bm{\widehat\varrho}^{(s)}}
\le
S(A^m|R)_{\varrho^{(m)}} .
\ee

Equivalently, when rates are normalized per original tensor factor rather
than per \(m\)-block,
\bb
\overline S
(A|\widetilde R R)_{\bm{\widehat\varrho}^{(s)}}
\le
\frac1m
S(A^m|R)_{\varrho^{(m)}} .
\ee

Here \((lm+s)/l\to m\) for fixed \(m\) and \(s\), which
accounts for the factor \(1/m\).

Since there are only finitely many possible remainders
\(s\in\{0,\ldots,m-1\}\), we obtain a cq-extension sequence for \(\bm\rho_{AB}\), with pure
conditional states, by using, at blocklength \(n\), the
construction corresponding to the decomposition \(n=lm+s\).
For the finitely many blocklengths \(n<m\), choose arbitrary
cq-extensions with pure conditional states; these choices do
not affect the asymptotic spectral entropy rates.

The spectral sup-rate of a sequence obtained by interleaving finitely many
subsequences is bounded above by the maximum of the spectral sup-rates of
those subsequences. Since the preceding bound is independent of \(s\), the
resulting cq-extension sequence \(\bm{\widehat\varrho}\) satisfies
\bb
\overline S(A|\widetilde R R)_{\bm{\widehat\varrho}}
\le
\frac1m
S(A^m|R)_{\varrho^{(m)}} .
\ee

The information-spectrum formula for entanglement cost therefore gives
\bb
E_C(\bm\rho_{AB})
\le
\frac1m
S(A^m|R)_{\varrho^{(m)}} .
\ee

By the choice of the ensemble,
\bb
E_C(\bm\rho_{AB})
\le
\frac1m
E_F(\rho_{AB}^{\otimes m})
+
\frac{\eta}{m}.
\ee

Since \(\eta>0\) was arbitrary, we obtain
\bb
E_C(\bm\rho_{AB})
\le
\frac1m
E_F(\rho_{AB}^{\otimes m}) .
\ee

Finally, taking the infimum over \(m\ge1\) yields
\bb
E_C(\bm\rho_{AB})
\le
\inf_{m\ge1}
\frac1m
E_F(\rho_{AB}^{\otimes m})
=
E_F^\infty(\rho_{AB}).
\ee

Thus every rate \(R>E_F^\infty(\rho_{AB})\) is achievable for
entanglement dilution of \(\bm\rho_{AB}\).

For a prescribed defect-size sequence \((r_n)_n\), the
preceding argument applies to every MSR almost-iid target
sequence along \(\rho_{AB}\) with defect size at most \(r_n\)
at blocklength \(n\).

Consequently, the above argument gives
\bb
\sup_{\bm\rho_{AB}}
E_C(\bm\rho_{AB})
\le
E_F^\infty(\rho_{AB}),
\ee
where the supremum is over all MSR almost-iid target sequences along
\(\rho_{AB}\) with defect size at most \(r_n\) at blocklength \(n\).

This is a common upper bound on the individual entanglement costs.
The cq-extensions and the associated dilution protocols may depend
on the particular target sequence.
\end{proof}
}

{\begin{thm}[(Asymptotic reversibility for MSR almost-iid sources
along a pure reference state)]
\label{thm:pure-msr-dilution-converse}
Let \(\ket{\phi}_{AB}\in\mathcal H_A\otimes\mathcal H_B\) be a
bipartite pure state, and set
\(\phi_{AB}:=\ket{\phi}\!\bra{\phi}_{AB}\) and
\(\phi_A:=\Tr_B\phi_{AB}\).
Let \((r_n)_n\) be a sequence of integers satisfying
\(0\leq r_n\leq n\) and \(r_n=o(n)\).
For every MSR almost-iid sequence
\(\bm\omega_{AB}:=(\omega_{A^nB^n})_n\) along \(\phi_{AB}\),
with defect size at most \(r_n\) at blocklength \(n\),
\bb
E_D(\bm\omega_{AB})
=
E_C(\bm\omega_{AB})
=
S(\phi_A).
\ee
The states \(\omega_{A^nB^n}\) may be mixed.
\end{thm}

\begin{proof}
For a pure reference state,
\(E_D(\phi_{AB})=E_F^\infty(\phi_{AB})=S(\phi_A)\).
Theorems~\ref{thm:MSR-universal-distillation} and
\ref{thm:MSR-dilution} therefore give
\bb
S(\phi_A)\leq E_D(\bm\omega_{AB}),
\qquad
E_C(\bm\omega_{AB})\leq S(\phi_A).
\ee
It remains to verify
\(E_D(\bm\omega_{AB})\leq E_C(\bm\omega_{AB})\).
The following argument applies to arbitrary sequences of bipartite states.

Suppose, to the contrary, that
\(E_D(\bm\omega_{AB})>E_C(\bm\omega_{AB})\), and choose rates
\bb
E_C(\bm\omega_{AB})<C<R<E_D(\bm\omega_{AB}).
\ee
Write \(\omega_n:=\omega_{A^nB^n}\).
By the operational definitions, there exist LOCC maps
\(\mathcal E_n\) and \(\mathcal D_n\), and positive integers
\(L_n\) and \(M_n\), such that
\bb
\limsup_{n\to\infty}\frac1n\log L_n\leq C,
\qquad
\liminf_{n\to\infty}\frac1n\log M_n\geq R,
\ee
and
\bb
F\bigl(\mathcal E_n(\Phi_{L_n}),\omega_n\bigr)\longrightarrow1,
\qquad
F\bigl(\mathcal D_n(\omega_n),\Phi_{M_n}\bigr)\longrightarrow1.
\ee
Set
\(\sigma_n:=(\mathcal D_n\circ\mathcal E_n)(\Phi_{L_n})\).
Using the fidelity--trace-distance inequalities, contractivity of trace
distance, and its triangle inequality, we obtain
\bb
\begin{aligned}
T(\sigma_n,\Phi_{M_n})
&\leq
T\bigl(\mathcal D_n(\mathcal E_n(\Phi_{L_n})),
\mathcal D_n(\omega_n)\bigr)
+T\bigl(\mathcal D_n(\omega_n),\Phi_{M_n}\bigr)
\\
&\leq
T\bigl(\mathcal E_n(\Phi_{L_n}),\omega_n\bigr)
+T\bigl(\mathcal D_n(\omega_n),\Phi_{M_n}\bigr)
\longrightarrow0.
\end{aligned}
\ee
Hence \(F(\sigma_n,\Phi_{M_n})\to1\).

On the other hand, an LOCC map acting on \(\Phi_{L_n}\) produces a
convex mixture of pure states of Schmidt rank at most \(L_n\).
By the Cauchy--Schwarz inequality, the squared overlap of any such pure
state with \(\Phi_{M_n}\) is at most
\(\min\{1,L_n/M_n\}\).
Since squared fidelity with a pure target is linear in the other state,
\bb
F(\sigma_n,\Phi_{M_n})^2
=\Tr(\Phi_{M_n}\sigma_n)
\leq\min\{1,L_n/M_n\}.
\ee
Fix \(0<\varepsilon<(R-C)/2\).
The rate bounds imply that, for all sufficiently large \(n\),
\bb
\frac{L_n}{M_n}
\leq
2^{-n(R-C-2\varepsilon)}
\longrightarrow0.
\ee
This contradicts \(F(\sigma_n,\Phi_{M_n})\to1\).
Thus \(E_D(\bm\omega_{AB})\leq E_C(\bm\omega_{AB})\), and
\bb
S(\phi_A)
\leq E_D(\bm\omega_{AB})
\leq E_C(\bm\omega_{AB})
\leq S(\phi_A).
\ee
The claimed equalities follow.
\end{proof}}

{\begin{cor}[(Normalized block entanglement bounds)]
\label{cor:msr-block-entanglement-bounds}
Let \(\rho_{AB}\) be a bipartite state on finite-dimensional
Hilbert spaces, and let
\(\boldsymbol{\omega}_{AB}:=(\omega_{A^nB^n})_{n\geq1}\)
be an MSR almost-iid sequence along \(\rho_{AB}\), with
defect sizes \(r_n=o(n)\). Then
\bb
\begin{aligned}
\liminf_{n\to\infty}
\frac1n E_D(\omega_{A^nB^n})
&\geq E_D(\rho_{AB}),\\
\limsup_{n\to\infty}
\frac1n E_C(\omega_{A^nB^n})
&\leq E_C(\rho_{AB})
=E_F^\infty(\rho_{AB}).
\end{aligned}
\ee
Here \(E_D(\omega_{A^nB^n})\) and
\(E_C(\omega_{A^nB^n})\) denote the standard iid
entanglement measures of the individual bipartite block state.

In particular, if
\(\rho_{AB}=\phi_{AB}:=\ket{\phi}\!\bra{\phi}_{AB}\)
is pure, then
\bb
\lim_{n\to\infty}
\frac1n E_D(\omega_{A^nB^n})
=
\lim_{n\to\infty}
\frac1n E_C(\omega_{A^nB^n})
=
S(\phi_A),
\ee
where \(\phi_A:=\Tr_B\phi_{AB}\).
\end{cor}

\begin{proof}
Write \(\omega_n:=\omega_{A^nB^n}\).
If \(E_D(\rho_{AB})=0\), the distillation lower bound
follows from nonnegativity. Otherwise, fix
\(0<R<E_D(\rho_{AB})\).
Our Theorem~\ref{thm:MSR-universal-distillation}
provides LOCC maps \(\mathcal D_n\) such that
\bb
\zeta_n:=\mathcal D_n(\omega_n),
\qquad
T(\zeta_n,\Phi_{M_n})\longrightarrow0,
\qquad
M_n:=\lfloor2^{nR}\rfloor.
\ee
Let \(A'_n\) and \(B'_n\) denote the output systems,
each of dimension \(M_n\). Since
\(\log M_n=O(n)\), continuity of the von Neumann entropy~\cite{Audenaert2007}
gives
\bb
S((\zeta_n)_{B'_n})-S(\zeta_n)
=
\log M_n+o(n).
\ee
By LOCC monotonicity and the hashing inequality
of~\cite{DevetakWinter2005},
\bb
E_D(\omega_n)
\geq E_D(\zeta_n)
\geq S((\zeta_n)_{B'_n})-S(\zeta_n).
\ee
Dividing by \(n\), taking the liminf, and then letting
\(R\uparrow E_D(\rho_{AB})\) proves the first bound.

For the cost upper bound, fix
\(C>E_F^\infty(\rho_{AB})\).
Our Theorem~\ref{thm:MSR-dilution}
provides LOCC maps \(\mathcal E_n\) and integers
\(L_n\) such that
\bb
\xi_n:=\mathcal E_n(\Phi_{L_n}),
\qquad
T(\xi_n,\omega_n)\longrightarrow0,
\qquad
\limsup_{n\to\infty}\frac1n\log L_n\leq C.
\ee
Asymptotic continuity of \(E_F\), together with
\(\log\dim(\mathcal H_A^{\otimes n}\otimes
\mathcal H_B^{\otimes n})=O(n)\), implies
\bb
E_F(\omega_n)\leq E_F(\xi_n)+o(n).
\ee
Using \(E_C=E_F^\infty\) and LOCC monotonicity
of \(E_F\)~\cite{HaydenHorodeckiTerhal2001}, we obtain
\bb
E_C(\omega_n)
=
E_F^\infty(\omega_n)
\leq E_F(\omega_n)
\leq E_F(\xi_n)+o(n)
\leq\log L_n+o(n).
\ee
Dividing by \(n\), taking the limsup, and then letting
\(C\downarrow E_F^\infty(\rho_{AB})\) proves the
second bound.

Finally, for a pure reference state,
\(E_D(\phi_{AB})=E_C(\phi_{AB})=S(\phi_A)\).
The two bounds, together with
\(E_D(\omega_n)\leq E_C(\omega_n)\) for every \(n\),
give both asserted limits.
\end{proof}}

The following lemma concerns defects in fixed, known positions.
\begin{lemma}[(Entanglement-cost lower bound for fixed identical defects)]
\label{lem:fixed-identical-defects-cost-lower-bound}
Let \(\rho_{AB},\sigma_{AB}\in\mathcal D(\mathcal H_A\otimes\mathcal H_B)\).
Let \((r_n)_n\) be a sequence of integers satisfying
\(0\leq r_n\leq n\) and \(r_n=o(n)\). Define
\bb
\omega^{(n)}_{A^nB^n}
:=
\rho_{AB}^{\otimes(n-r_n)}
\otimes
\sigma_{AB}^{\otimes r_n}.
\ee
Then the entanglement cost of the sequence
\(\bm\omega_{AB}:=(\omega^{(n)}_{A^nB^n})_n\) satisfies
\bb
E_C(\bm\omega_{AB})\ge E_F^\infty(\rho_{AB}).
\ee
\end{lemma}

\begin{proof}
Let \(R\) be an achievable entanglement dilution rate for
\(\bm\omega_{AB}\). Then there exist integers \(M_n\) and LOCC maps
\(\Lambda_n\) such that
\bb
\limsup_{n\to\infty}\frac1n\log M_n\le R
\ee
and
\bb
F^2\bigl(\Lambda_n(\Phi_{M_n}),\omega^{(n)}_{A^nB^n}\bigr)\to 1.
\ee
After applying the dilution protocol, Alice and Bob may locally discard the
last \(r_n\) bipartite systems. Since the partial trace is an LOCC operation
and fidelity is monotone under quantum channels, the resulting protocol
prepares \(\rho_{AB}^{\otimes(n-r_n)}\) with fidelity tending to one.

Set \(m_n:=n-r_n\), so that \(m_n/n\to1\) and \(m_n\to\infty\),
and let
\bb
\zeta_n:=
\Tr_{A_{m_n+1}^nB_{m_n+1}^n}
\Lambda_n(\Phi_{M_n}).
\ee
Then \(F(\zeta_n,\rho_{AB}^{\otimes m_n})\to1\).
Asymptotic continuity of the entanglement of formation,
as used in the iid converse~\cite{HaydenHorodeckiTerhal2001},
gives
\bb
E_F(\zeta_n)
\geq E_F(\rho_{AB}^{\otimes m_n})-b_n,
\qquad
b_n\geq0,\quad b_n=o(m_n).
\ee
Here the logarithm of the dimension of the output Hilbert space
is \(m_n\log(d_Ad_B)=O(m_n)\), so vanishing approximation error
gives a sublinear continuity correction.
LOCC monotonicity of \(E_F\) yields
\bb
\log M_n=E_F(\Phi_{M_n})
\geq E_F(\zeta_n)
\geq E_F(\rho_{AB}^{\otimes m_n})-b_n.
\ee
Dividing by \(m_n\) for sufficiently large \(n\), and using
the definition of \(E_F^\infty\), we obtain
\bb
\liminf_{n\to\infty}\frac1{m_n}\log M_n
\geq E_F^\infty(\rho_{AB}).
\ee
Since \(m_n/n\to1\), it follows that
\bb
R
\ge
\liminf_{n\to\infty}\frac1n\log M_n
\ge
E_F^\infty(\rho_{AB}).
\ee
Since \(R\) was arbitrary among achievable rates, we conclude that
\bb
E_C(\bm\omega_{AB})\ge E_F^\infty(\rho_{AB}).
\ee
This proves the claim.
\end{proof}
The preceding lemma concerns defects in fixed, known positions.
If \(\rho_{AB}\neq\sigma_{AB}\) and \(0<r_n<n\), the
single-site marginals of \(\omega_{A^nB^n}^{(n)}\) are
\(\rho_{AB}\) at the reference positions and \(\sigma_{AB}\)
at the defect positions. Hence this state is not invariant
under simultaneous permutations of the pairs \((A_i,B_i)\).
Every MSR almost-iid state, however, is permutation invariant, since
it is the marginal of a permutation-invariant extension.
Thus the states in the lemma are generally outside the
MSR almost-iid class, and its lower bound does not establish an
entanglement-dilution converse for arbitrary mixed MSR almost-iid sources.

Averaging over permutations does not extend the proof to the
averaged state: the discarding step relies on knowing the
defect positions, and discarding a fixed set of \(r_n\)
bipartite systems from the averaged state need not leave
\(\rho_{AB}^{\otimes(n-r_n)}\).
}

\section{Conclusions}\label{sec:conclude}

The results of this paper show that the MSR almost-iid model provides
a robust framework for entanglement manipulation beyond the idealized
tensor-power setting. The structural properties of MSR almost-iid
sources, in particular their stability under local tensor-power
channels, marginals and blocking operations, make it possible to
transfer entanglement-manipulation arguments from the iid setting to
sources with sublinear deviations from tensor-power structure.
The associated entropy-rigidity results provide the spectral entropy
bounds needed for these operational arguments.

Combining these structural and entropic ingredients, we established
universal entanglement concentration for pure MSR almost-iid sources,
universal entanglement distillation for general MSR almost-iid sources,
and entanglement dilution for every prescribed MSR almost-iid target
sequence. For any MSR almost-iid sequence \(\bm\omega_{AB}\) along
a reference state \(\rho_{AB}\), with a prescribed defect bound
\(r_n=o(n)\), our results give
\bb
E_D(\bm\omega_{AB})\geq E_D(\rho_{AB}),
\qquad
E_C(\bm\omega_{AB})\leq E_F^\infty(\rho_{AB}).
\ee
At every rate strictly below \(E_D(\rho_{AB})\), the distillation
protocol has an exponentially vanishing error uniformly over the
corresponding MSR almost-iid class. This protocol may depend on the
reference state, but is independent of the particular source.
The dilution protocol may depend on the prescribed target sequence.
For pure source states, the Schur concentration protocol can
additionally be chosen independently of the reference state.

The distillation argument also applies when classical communication
is restricted to messages from Alice to Bob. Let
$E_D^{\rightarrow}$ denote distillable entanglement under this
restriction. For every rate
$R<E_D^{\rightarrow}(\rho_{AB})$, the initial finite-block protocol
$\Gamma$ can be chosen to be one-way. Twirling and hashing can also
be implemented with classical communication only from Alice to Bob:
Alice communicates her random choices and measurement outcomes,
while Bob performs the decoding and correction. The operator
domination and error-transfer arguments are unchanged. Consequently,
a single sequence of one-way protocols achieves every such rate,
with exponentially vanishing error uniformly over each prescribed
MSR almost-iid class with $r_n=o(n)$. In particular,
\bb
E_D^{\rightarrow}(\boldsymbol{\omega}_{AB})
\ge E_D^{\rightarrow}(\rho_{AB}).
\ee
The protocols depend only on the reference state, the chosen rate,
and the local Hilbert spaces, and can be chosen independently of
the particular source and the sequence $(r_n)_n$.

When the reference state is pure,
\(\phi_{AB}:=\ket{\phi}\!\bra{\phi}_{AB}\), the distillation and
dilution bounds imply asymptotic reversibility:
\bb
E_D(\bm\omega_{AB})
=
E_C(\bm\omega_{AB})
=
S(\phi_A),
\qquad
\phi_A:=\Tr_B\phi_{AB}.
\ee
As established in
Theorem~\ref{thm:pure-msr-dilution-converse}, this equality holds
for every MSR almost-iid sequence along \(\phi_{AB}\) with
\(r_n=o(n)\), including sequences whose individual states are mixed.

\subsection*{Acknowledgments}
ND is grateful to her former collaborator Garry Bowen, who introduced her
to the information-spectrum framework of quantum information theory. She fondly remembers their many interesting collaborations, which started her ``beyond
i.i.d.'' journey. She is grateful to Bjarne Bergh and Liuhang Ye for their comments on the first version of this paper. She thanks Ludovico Lami for sharing his notes on the
stability of entanglement measures. She also thanks Ryuji Takagi on sharing his initial thoughts of approaching the entanglement distillation problem via a different method - using state merging. She is also grateful to him for pointing out that the result on entanglement distillation can be extended to the case of one-way LOCC from Alice to Bob. 
She is supported by the Engineering and Physical Sciences Research Council
[Grant Ref: EP/Y028732/1].
\smallskip

\noindent
{\em{AI statement}}.— The first version of this paper was written  without using generative AI. In this version, the author used results obtained by Tirone et al.~to prove a new result, namely, universal entanglement distillation from arbitrary MSR almost-iid sources. Then ChatGPT 6.1 was used to check this proof and all other mathematical statements and proofs in the paper, and to improve and streamline the presentation.  The author carefully verified all AI-assisted refinements and takes full responsibility for the content of this work.
\bibliography{biblio}

@misc{Tirone26,
  author        = {Tirone, Salvatore and Mele, Francesco Anna
                   and Giovannetti, Vittorio and Lami, Ludovico},
  title         = {Universal entanglement distillation},
  year          = {2026},
  eprint        = {2609.40080},
  archivePrefix = {arXiv},
  primaryClass  = {quant-ph},
  url           = {https://arxiv.org/abs/2609.40080}
}

@misc{MazzolaRenner2026Operational,
  author        = {Mazzola, Giulia and Renner, Renato},
  title         = {Operational transformation rates of quantum states},
  year          = {2026},
  eprint        = {2610.02025},
  archivePrefix = {arXiv},
  primaryClass  = {quant-ph},
  url           = {https://arxiv.org/abs/2610.02025}
}

@article{bowen2008asymptotic,
  title={Asymptotic entanglement manipulation of bipartite pure states},
  author={Bowen, Garry and Datta, Nilanjana},
  journal={IEEE transactions on information theory},
  volume={54},
  number={8},
  pages={3677--3686},
  year={2008},
  publisher={IEEE}
}

@inproceedings{bowen2006beyond,
  title={Beyond iid in quantum information theory},
  author={Bowen, Garry and Datta, Nilanjana},
  booktitle={2006 IEEE International Symposium on Information Theory},
  pages={451--455},
  year={2006},
  organization={IEEE}
}

@article{datta2026entropy,
  title={Entropy Concentration and Universal Typicality for Weakly Almost iid Quantum Sources},
  author={Datta, Nilanjana},
  journal={arXiv preprint arXiv:2605.20092},
  year={2026}
}

@article{HayashiNagaoka2003,
  author = {Masahito Hayashi and Hiroshi Nagaoka},
  title = {General formulas for capacity of classical-quantum channels},
  journal = {IEEE Transactions on Information Theory},
  volume = {49},
  number = {7},
  pages = {1753--1768},
  year = {2003},
  doi = {10.1109/TIT.2003.813556}
}

@book{Tomamichel2016,
  author = {Marco Tomamichel},
  title = {Quantum Information Processing with Finite Resources: Mathematical Foundations},
  publisher = {Springer},
  year = {2016}
}

@article{girardi2026new,
  title={New approaches to almost iid information theory},
  author={Girardi, Filippo and De Palma, Giacomo and Lami, Ludovico},
  journal={arXiv preprint arXiv:2605.15114},
  year={2026}
}

@article{Nielsen1999,
  author  = {Michael A. Nielsen},
  title   = {Conditions for a Class of Entanglement Transformations},
  journal = {Physical Review Letters},
  volume  = {83},
  number  = {2},
  pages   = {436--439},
  year    = {1999}
}

@article{BennettBernsteinPopescuSchumacher1996,
  author  = {Bennett, Charles H. and Bernstein, Herbert J. and Popescu, Sandu and Schumacher, Benjamin},
  title   = {Concentrating partial entanglement by local operations},
  journal = {Physical Review A},
  volume  = {53},
  number  = {4},
  pages   = {2046--2052},
  year    = {1996},
  doi     = {10.1103/PhysRevA.53.2046},
  eprint  = {quant-ph/9511030},
  archivePrefix = {arXiv}
}

@article{LoPopescu1999,
  author  = {Lo, Hoi-Kwong and Popescu, Sandu},
  title   = {Concentrating entanglement by local actions: Beyond mean values},
  journal = {Physical Review A},
  volume  = {63},
  number  = {2},
  pages   = {022301},
  year    = {2001},
  doi     = {10.1103/PhysRevA.63.022301},
  eprint  = {quant-ph/9707038},
  archivePrefix = {arXiv}
}

@article{BennettDiVincenzoSmolinWootters1996,
  author  = {Bennett, Charles H. and DiVincenzo, David P. and Smolin, John A. and Wootters, William K.},
  title   = {Mixed-state entanglement and quantum error correction},
  journal = {Physical Review A},
  volume  = {54},
  number  = {5},
  pages   = {3824--3851},
  year    = {1996},
  doi     = {10.1103/PhysRevA.54.3824},
  eprint  = {quant-ph/9604024},
  archivePrefix = {arXiv}
}

@article{HaydenHorodeckiTerhal2001,
  author  = {Hayden, Patrick M. and Horodecki, Michal and Terhal, Barbara M.},
  title   = {The asymptotic entanglement cost of preparing a quantum state},
  journal = {Journal of Physics A: Mathematical and General},
  volume  = {34},
  number  = {35},
  pages   = {6891--6898},
  year    = {2001},
  doi     = {10.1088/0305-4470/34/35/314},
  eprint  = {quant-ph/0008134},
  archivePrefix = {arXiv}
}

@article{KumagaiHayashi2013,
  author        = {Kumagai, Wataru and Hayashi, Masahito},
  title         = {Second-Order Asymptotics of Conversions of Distributions and Entangled States Based on Rayleigh-Normal Probability Distributions},
  journal       = {IEEE Transactions on Information Theory},
  volume        = {63},
  number        = {3},
  pages         = {1829--1857},
  year          = {2017},
  doi           = {10.1109/TIT.2016.2645223},
  eprint        = {1306.4166},
  archivePrefix = {arXiv},
  primaryClass  = {quant-ph}
}

@article{DattaLeditzky2014,
  author        = {Datta, Nilanjana and Leditzky, Felix},
  title         = {Second-order asymptotics for source coding, dense coding and pure-state entanglement conversions},
  journal       = {IEEE Transactions on Information Theory},
  volume        = {61},
  number        = {1},
  pages         = {582--608},
  year          = {2015},
  doi           = {10.1109/TIT.2014.2366994},
  eprint        = {1403.2543},
  archivePrefix = {arXiv},
  primaryClass  = {quant-ph}
}

@article{DevetakWinter2005,
  author  = {Devetak, Igor and Winter, Andreas},
  title   = {Distillation of secret key and entanglement from quantum states},
  journal = {Proceedings of the Royal Society A: Mathematical, Physical and Engineering Sciences},
  volume  = {461},
  number  = {2053},
  pages   = {207--235},
  year    = {2005},
  doi     = {10.1098/rspa.2004.1372},
  eprint  = {quant-ph/0306078},
  archivePrefix = {arXiv}
}

@article{FuchsVanDeGraaf1999,
  author        = {Fuchs, Christopher A. and van de Graaf, Jeroen},
  title         = {Cryptographic distinguishability measures for quantum-mechanical states},
  journal       = {IEEE Transactions on Information Theory},
  volume        = {45},
  number        = {4},
  pages         = {1216--1227},
  year          = {1999},
  doi           = {10.1109/18.761271},
  eprint        = {quant-ph/9712042},
  archivePrefix = {arXiv}
}

@article{FangWangTomamichelDuan2017,
  author        = {Fang, Kun and Wang, Xin and Tomamichel, Marco and Duan, Runyao},
  title         = {Non-asymptotic entanglement distillation},
  journal       = {IEEE Transactions on Information Theory},
  volume        = {65},
  number        = {10},
  pages         = {6454--6465},
  year          = {2019},
  doi           = {10.1109/TIT.2019.2914688},
  eprint        = {1706.06221},
  archivePrefix = {arXiv},
  primaryClass  = {quant-ph}
}

@article{bowen2011entanglement,
  title={Entanglement cost for sequences of arbitrary quantum states},
  author={Bowen, Garry and Datta, Nilanjana},
  journal={Journal of Physics A: Mathematical and Theoretical},
  volume={44},
  number={4},
  pages={045302},
  year={2011}
}

@article{girardi2026quantum,
  title={Quantum Shannon theory made robust: a tale of three protocols for almost iid sources},
  author={Girardi, Filippo and Datta, Nilanjana and De Palma, Giacomo and Lami, Ludovico},
  journal={arXiv preprint arXiv:2605.18726},
  year={2026}
}

@article{Datta_2009,
   title={Smooth {R}\'enyi Entropies and the Quantum Information Spectrum},
   volume={55},
   ISSN={0018-9448},
   url={http://dx.doi.org/10.1109/TIT.2009.2018340},
   DOI={10.1109/tit.2009.2018340},
   number={6},
   journal={IEEE Transactions on Information Theory},
   publisher={Institute of Electrical and Electronics Engineers (IEEE)},
   author={Datta, Nilanjana and Renner, Renato},
   year={2009},
   month=June, pages={2807–2815} }

@article{mazzola2026almost,
  title={Almost-iid information theory},
  author={Mazzola, Giulia and Sutter, David and Renner, Renato},
  journal={arXiv preprint arXiv:2603.15792},
  year={2026}
}

@phdthesis{Renner2005,
  author = {Renato Renner},
  title = {Security of Quantum Key Distribution},
  school = {ETH Zurich},
  year = {2005},
  note = {arXiv:quant-ph/0512258}
}

@inproceedings{BowenDattaISIT2006,
  author    = {Garry Bowen and Nilanjana Datta},
  title     = {Asymptotic Quantum Coding Theorems for Bipartite Resources},
  booktitle = {Proceedings of the International Symposium on Information Theory (ISIT 2006)},
  pages     = {451--455},
  year      = {2006},
  doi       = {10.1109/ISIT.2006.261735},
  eprint    = {quant-ph/0604013},
  archivePrefix = {arXiv},
  primaryClass = {quant-ph}
}

@article{hayashi2002universal,
  title={Universal distortion-free entanglement concentration},
  author={Hayashi, Masahito and Matsumoto, Keiji},
  journal={arXiv preprint quant-ph/0209030},
  year={2002}
}

@article{PhysRevA.75.062338,
  title = {Universal distortion-free entanglement concentration},
  author = {Matsumoto, Keiji and Hayashi, Masahito},
  journal = {Phys. Rev. A},
  volume = {75},
  issue = {6},
  pages = {062338},
  numpages = {13},
  year = {2007},
  month = {Jun},
  publisher = {American Physical Society},
  doi = {10.1103/PhysRevA.75.062338},
  url = {https://link.aps.org/doi/10.1103/PhysRevA.75.062338}
}

@misc{Matsumoto2007,
  author        = {Matsumoto, Keiji},
  title         = {Entanglement Cost and Distillable Entanglement of
                   Symmetric States},
  year          = {2007},
  eprint        = {0708.3129},
  archivePrefix = {arXiv},
  primaryClass  = {quant-ph},
  doi           = {10.48550/arXiv.0708.3129}
}

@article{BowenDatta2006,
  author        = {Garry Bowen and Nilanjana Datta},
  title         = {Quantum Coding Theorems for Arbitrary Sources, Channels and Entanglement Resources},
  year          = {2006},
  archivePrefix = {arXiv},
  eprint        = {quant-ph/0610003},
  primaryClass  = {quant-ph},
  note          = {arXiv:quant-ph/0610003}
}

@book{Hayashi2017,
  author    = {Masahito Hayashi},
  title     = {Quantum Information Theory: Mathematical Foundation},
  publisher = {Springer},
  year      = {2017},
  series    = {Graduate Texts in Physics},
  doi       = {10.1007/978-3-662-49725-8}
}

@article{Hayashi2006,
  author={Hayashi, M.},
  journal={IEEE Trans. Inf. Theory}, 
  title={General formulas for fixed-length quantum entanglement concentration}, 
  year={2006},
  volume={52},
  number={5},
  pages={1904--1921},
  doi={10.1109/TIT.2006.872976}
}

@article{Datta08,
  author={Datta, N.},
  journal={IEEE Trans. Inf. Theory}, 
  title={Min- and max-relative entropies and a new entanglement monotone}, 
  year={2009},
  volume={55},
  number={6},
  pages={2816--2826},
  doi={10.1109/TIT.2009.2018325}
}

@article{BrandaoEisert2008,
  author  = {Brand{\~a}o, Fernando G. S. L. and Eisert, Jens},
  title   = {Correlated Entanglement Distillation and the Structure of
             the Set of Undistillable States},
  journal = {Journal of Mathematical Physics},
  volume  = {49},
  number  = {4},
  pages   = {042102},
  year    = {2008},
  doi     = {10.1063/1.2888925}
}

@article{BBM,
  title = {Quantum cryptography without {B}ell's theorem},
  author = {Bennett, C. H. and Brassard, G. and Mermin, N. D.},
  journal = {Phys. Rev. Lett.},
  volume = {68},
  issue = {5},
  pages = {557--559},
  numpages = {0},
  year = {1992},
  publisher = {American Physical Society},
  doi = {10.1103/PhysRevLett.68.557},
  url = {https://link.aps.org/doi/10.1103/PhysRevLett.68.557}
}

@article{Fawzi-Renner,
author="Fawzi, O. and Renner, R.",
title="Quantum Conditional Mutual Information and Approximate {M}arkov Chains",
journal="Commun. Math. Phys.",
year="2015",
volume="340",
number="2",
pages="575--611"
}

@article{Audenaert2007,
  author={Audenaert, K. M. R.},
  title={A sharp continuity estimate for the von {N}eumann entropy},
  journal={J. Phys. A},
  volume={40},
  number={28},
  pages={8127},
  year={2007}
}
\appendix
\section{Information-spectrum preliminaries}\label{app:info-spec}
Note that if \(R>\overline S(\bm\rho)\), then
\(-R<\underline D(\bm\rho\|\bm I\), since \(\overline S(\bm\rho)=-\underline D(\bm\rho\|\bm I)\). Hence we may choose
\(\gamma\in\mathbb R\) such that
\(-R<\gamma<\underline D(\bm\rho\|\bm I)\). By the definition of
\(\underline D(\bm\rho\|\bm I\),
\bb
\Tr\!\left[
\{\rho_n-2^{n\gamma}I_n>0\}\rho_n
\right]\to 1 \quad \text{as} \quad n \to \infty.
\ee
Let \(\Pi_n:=\{\rho_n-2^{n\gamma} I_n>0\}\). Then
\(\Tr(\Pi_n\rho_n)\to1\) as $n \to \infty$. Moreover, on the range of \(\Pi_n\), all
eigenvalues of \(\rho_n\) are strictly larger than \(2^{n\gamma}\).
Since \(\Tr\rho_n=1\), this gives \(2^{n\gamma}\Tr\Pi_n\le1\), and
therefore \(\Tr\Pi_n\le2^{-n\gamma}\). Since \(\gamma>-R\), we have
\(\Tr\Pi_n\le2^{nR}\) for all sufficiently large \(n\). Thus every rate
strictly larger than \(\overline S(\bm\rho)\) admits such a sequence of
high-probability projections.

Conversely, suppose that, for some \(R\in\mathbb R\), there exists a
sequence of projections \((\Pi_n)_n\) such that
\(\Tr(\Pi_n\rho_n)\to1\) and \(\Tr\Pi_n\le2^{nR}\) for all sufficiently
large \(n\). Fix \(\delta>0\), and define
\(Q_n:=\{\rho_n-2^{-n(R+\delta)}I_n>0\}\). On the complement of
\(Q_n\), we have \(\rho_n\le2^{-n(R+\delta)}I_n\). Hence
\bb
\Tr\!\left[\Pi_n(I_n-Q_n)\rho_n\right]
\le
2^{-n(R+\delta)}\Tr\Pi_n
\le
2^{-n\delta}
\ee
for all sufficiently large \(n\). It follows that
\(\Tr(Q_n\rho_n)\ge \Tr(\Pi_n\rho_n)-2^{-n\delta}\to1\), or
equivalently,
\bb
\Tr\!\left[
\{\rho_n-2^{-n(R+\delta)}I_n>0\}\rho_n
\right]\to 1 \quad \text{as} \quad n \to \infty.
\ee
By the definition of \(\underline D(\bm\rho\|\bm I)\), this implies
\(-R-\delta\le \underline D(\bm\rho\|\bm I)\). Since \(\delta>0\) is
arbitrary, we obtain
\(\overline S(\bm\rho)=-\underline D(\bm\rho\|\bm I)\le R\). Taking
the infimum over all such \(R\) gives the claimed characterization of
$\overline S(\bm\rho)$.
\section{Technical properties of MSR almost-iid sources}
\subsection{Proof of Lemma~\ref{lem:msr-stability-local-channels}}
\label{app:structure-msr-stability}
\begin{proof}
Let \(V:A\to A'F\) be a Stinespring isometry for \(\Lambda_A\), so that
\(
\Lambda_A(\cdot)=\Tr_F[V(\cdot)V^\dagger]
\).
Choose a purification \(\ket{\theta}_{ABE}\) of \(\rho_{AB}\). By
Remark~2.4(a) of~\cite{mazzola2026almost}, for this choice of
purification there exists an MSR extension
\(\rho_{A^nB^nE^n}\) of \(\rho_{A^nB^n}\) satisfying the MSR almost-iid
permutation-invariance and support conditions. Define
\bb
\ket{\theta'}_{A'BFE}
:=
(V\otimes\id_{BE})\ket{\theta}_{ABE}.
\ee
Then \(\ket{\theta'}_{A'BFE}\) is a purification of \(\omega_{A'B}\).

The chosen extension \(\rho_{A^nB^nE^n}\) has the following
properties:

\begin{enumerate}
\item \(\rho_{A^nB^nE^n}\) is permutation invariant under simultaneous permutations of the triples
\((A_i,B_i,E_i)\);

\item
\bb
\operatorname{supp}(\rho_{A^nB^nE^n})
\subseteq
\operatorname{span}
\mathcal V
\bigl(
\mathcal H_{ABE}^{\otimes n},
\ket{\theta}_{ABE}^{\otimes(n-r_n)}
\bigr).
\ee
\end{enumerate}

Define
\bb
\widetilde\omega_{A'^nB^nF^nE^n}
:=
(V^{\otimes n}\otimes\id_{B^nE^n})
\rho_{A^nB^nE^n}
(V^{\otimes n}\otimes\id_{B^nE^n})^\dagger .
\ee
Tracing out \(F^nE^n\) yields
\bb
\Tr_{F^nE^n}
\widetilde\omega_{A'^nB^nF^nE^n}
=
\omega_{A'^nB^n},
\ee
so \(\widetilde\omega_{A'^nB^nF^nE^n}\) is an extension of
\(\omega_{A'^nB^n}\).

{As displayed below, this isometry intertwines the simultaneous
permutations of the input triples and the output quadruples. Hence,
\(\widetilde\omega_{A'^nB^nF^nE^n}\) is invariant under simultaneous
permutations of the output quadruples.}

{It remains to verify the support condition (i.e.\ the condition~\eqref{eq:supp} of Definition~\ref{def:MSR}). Since
\(\widetilde\omega_{A'^nB^nF^nE^n}\) is obtained from
\(\rho_{A^nB^nE^n}\) by applying the isometry
\(V^{\otimes n}\otimes I_{B^nE^n}\), it suffices to check how this
isometry acts on the spanning vectors of the original MSR almost-iid defect space.

Let
\bb
\ket{\Psi}
=
U_\pi^{ABE}
\bigl(
\ket{\theta}_{ABE}^{\otimes(n-r_n)}
\otimes
\ket{\Omega}_{ABE}^{(r_n)}
\bigr)
\ee
be one such spanning vector, where \(U_\pi^{ABE}\) denotes the simultaneous
permutation of the \(n\) triples \((A_i,B_i,E_i)\), and
\(\ket{\Omega}_{ABE}^{(r_n)}\) is an arbitrary vector on the remaining
\(r_n\) triples. Since \(V^{\otimes n}\) acts identically on each \(A_i\),
we have the intertwining relation
\bb
(V^{\otimes n}\otimes I_{B^nE^n})U_\pi^{ABE}
=
U_\pi^{A'BFE}
(V^{\otimes n}\otimes I_{B^nE^n}),
\ee
where \(U_\pi^{A'BFE}\) denotes the simultaneous permutation of the
\(n\) quadruples \((A'_i,B_i,F_i,E_i)\). Therefore
\bb
\begin{aligned}
(V^{\otimes n}\otimes I_{B^nE^n})\ket{\Psi}
&=
U_\pi^{A'BFE}
(V^{\otimes n}\otimes I_{B^nE^n})
\bigl(
\ket{\theta}_{ABE}^{\otimes(n-r_n)}
\otimes
\ket{\Omega}_{ABE}^{(r_n)}
\bigr)  \\
&=
U_\pi^{A'BFE}
\bigl(
\ket{\theta'}_{A'BFE}^{\otimes(n-r_n)}
\otimes
\ket{\Omega'}_{A'BFE}^{(r_n)}
\bigr),
\end{aligned}
\ee
where
\(\ket{\theta'}_{A'BFE}:=(V\otimes I_{BE})\ket{\theta}_{ABE}\), and
\[
\ket{\Omega'}_{A'BFE}^{(r_n)}
:=
(V^{\otimes r_n}\otimes I_{B^{r_n}E^{r_n}})
\ket{\Omega}_{ABE}^{(r_n)} .
\]
By linearity, the image of every spanning vector of
\(
\operatorname{span}
\mathcal V
\bigl(
\mathcal H_{ABE}^{\otimes n},
\ket{\theta}_{ABE}^{\otimes(n-r_n)}
\bigr)
\)
lies in
\noindent
{\[
\operatorname{span}
\mathcal V
\bigl(
\mathcal H_{A'BFE}^{\otimes n},
\ket{\theta'}_{A'BFE}^{\otimes(n-r_n)}
\bigr).
\]}
 Since
\(\rho_{A^nB^nE^n}\) is supported on the original defect space, its
isometric image satisfies
\bb
\operatorname{supp}
(\widetilde\omega_{A'^nB^nF^nE^n})
\subseteq
\operatorname{span}
\mathcal V
\bigl(
\mathcal H_{A'BFE}^{\otimes n},
\ket{\theta'}_{A'BFE}^{\otimes(n-r_n)}
\bigr).
\ee
Therefore \(\widetilde\omega_{A'^nB^nF^nE^n}\) is an MSR extension of
\(\omega_{A'^nB^n}\) along the purification
\(\ket{\theta'}_{A'BFE}\), with the same defect size \(r_n\). Since
\(r_n=o(n)\), the sequence \(\bm\omega_{A'B}\) is MSR almost-iid
along \(\omega_{A'B}\).}

\end{proof}

\section{Auxiliary information-spectrum estimates}\label{app:aux-info-spec}
Proof of Lemma~\ref{lem:finite-register-stability-conditional-spectrum}

{\begin{proof}
Data processing for the inf-spectral divergence rate under partial trace
(Proposition~4 of~\cite{bowen2006beyond}) gives
\bb
\underline D\!\left(
\bm\gamma_{ABZ}
\middle\|
(I_{A_l}\otimes\gamma_{B_lZ_l})_l
\right)
\geq
\underline D\!\left(
\bm\gamma_{AB}
\middle\|
(I_{A_l}\otimes\gamma_{B_l})_l
\right).
\ee
Consequently,
\bb
\overline S(A|BZ)_{\bm\gamma}
\leq
\overline S(A|B)_{\bm\gamma}.
\ee

Applying the same data-processing inequality under partial trace over
\(C_l\) yields
\bb
\underline D\!\left(
\bm\gamma_{ABCZ}
\middle\|
(I_{A_lC_l}\otimes\gamma_{B_lZ_l})_l
\right)
\geq
\underline D\!\left(
\bm\gamma_{ABZ}
\middle\|
(d_{C,l}I_{A_l}\otimes\gamma_{B_lZ_l})_l
\right),
\ee
since
\(\Tr_{C_l}(I_{A_lC_l}\otimes\gamma_{B_lZ_l})
=d_{C,l}I_{A_l}\otimes\gamma_{B_lZ_l}\).

Set
\bb
\omega_l:=\gamma_{A_lB_lZ_l},
\qquad
\sigma_l:=I_{A_l}\otimes\gamma_{B_lZ_l},
\qquad
a_l:=\frac1l\log d_{C,l},
\qquad
L:=\limsup_{l\to\infty}a_l.
\ee
We verify the scaling bound
\bb
\underline D\!\left(\bm\omega\middle\|(d_{C,l}\sigma_l)_l\right)
\geq
\underline D(\bm\omega\|\bm\sigma)-L.
\ee
For finite \(L\), fix any
\(\alpha<\underline D(\bm\omega\|\bm\sigma)-L\).
By the definition of the inf-spectral divergence rate, we can choose
\(\beta>\alpha+L\) such that the projections
\bb
P_l:=\{\omega_l-2^{l\beta}\sigma_l>0\}
\ee
satisfy \(\Tr(P_l\omega_l)\to1\). Let
\(\delta:=(\beta-\alpha-L)/2>0\).
For all sufficiently large \(l\),
\bb
a_l\leq L+\delta,
\qquad
\alpha+a_l-\beta\leq-\delta.
\ee
Moreover,
\bb
\Tr(P_l\sigma_l)
\leq 2^{-l\beta}\Tr(P_l\omega_l)
\leq 2^{-l\beta}.
\ee
Define
\bb
Q_l:=\{\omega_l-2^{l\alpha}d_{C,l}\sigma_l>0\}.
\ee
For every Hermitian operator \(X\) and every \(0\leq P\leq I\),
\(\Tr(PX)\leq\Tr(\{X>0\}X)\), by the spectral decomposition of
\(X\). Applying this inequality to
\(X=\omega_l-2^{l\alpha}d_{C,l}\sigma_l\) and \(P=P_l\), we obtain
\bb
\begin{aligned}
\Tr(Q_l\omega_l)
&\geq
\Tr\!\left[Q_l(\omega_l-2^{l\alpha}d_{C,l}\sigma_l)\right]
\\
&\geq
\Tr\!\left[P_l(\omega_l-2^{l\alpha}d_{C,l}\sigma_l)\right]
\\
&\geq
\Tr(P_l\omega_l)-2^{l(\alpha+a_l-\beta)}
\\
&\geq
\Tr(P_l\omega_l)-2^{-l\delta}
\longrightarrow1.
\end{aligned}
\ee
Since \(\Tr(Q_l\omega_l)\leq1\), the limit is one.
Thus \(\alpha\) is admissible in the definition of
\(\underline D(\bm\omega\|(d_{C,l}\sigma_l)_l)\).
Taking the supremum over such \(\alpha\) proves the scaling bound.
If \(\underline D(\bm\omega\|\bm\sigma)=-\infty\), this bound is
immediate for finite \(L\).

Combining the scaling bound with the two data-processing inequalities
gives
\bb
\begin{aligned}
\overline S(AC|BZ)_{\bm\gamma}
&\leq
-\underline D\!\left(
\bm\gamma_{ABZ}
\middle\|
(d_{C,l}I_{A_l}\otimes\gamma_{B_lZ_l})_l
\right)
\\
&\leq
\overline S(A|BZ)_{\bm\gamma}+L
\\
&\leq
\overline S(A|B)_{\bm\gamma}
+\limsup_{l\to\infty}\frac1l\log d_{C,l}.
\end{aligned}
\ee
When the right-hand side is \(+\infty\), the asserted upper bound is
immediate. In the applications here, \(\log d_{C,l}=o(l)\), so \(L=0\).
\end{proof}}

\section{Proofs of technical lemmas for entanglement dilution}\label{app:ent-dil}
\subsection{Proof of Lemma~\ref{lem:cq-lifting}}
\label{app:cq}

\begin{proof}
{Choose a minimal purification \(\ket{\theta}_{ABE}\) of
\(\rho_{AB}\), with
\(\dim\mathcal H_E=\operatorname{rank}\rho_{AB}\).
Since \(\varrho_{RAB}\) is a cq-extension of \(\rho_{AB}\) with
pure conditional states, the vector
\bb
\ket{\eta}_{ABR}
:=
\sum_i\sqrt{p_i}\,\ket{\phi_i}_{AB}\ket{i}_R
\ee
is also a purification of \(\rho_{AB}\). Since
\(\dim\mathcal H_R\geq\operatorname{rank}\rho_{AB}
=\dim\mathcal H_E\), there exists an isometry \(W:E\to R\) such that
\bb
(I_{AB}\otimes W)\ket{\theta}_{ABE}
=
\ket{\eta}_{ABR}.
\ee
Let \(\Delta_R\) denote the dephasing channel in the basis
\(\{\ket{i}_R\}_i\). Then
\bb
(\Delta_R\otimes\id_{AB})
\bigl(\ket{\eta}\!\bra{\eta}_{ABR}\bigr)
=
\varrho_{RAB}.
\ee
}

By Remark~2.4(a) of~\cite{mazzola2026almost}, for the chosen purification
\(\ket{\theta}_{ABE}\) there exists an MSR extension
\(\omega_{A^nB^nE^n}\) of \(\rho_{A^nB^n}\) which is invariant under
simultaneous permutations of the triples \((A_j,B_j,E_j)\) and satisfies
\bb
\operatorname{supp}\omega_{A^nB^nE^n}
\subseteq
\operatorname{span}{\mathcal V}
\bigl(
\mathcal H_{ABE}^{\otimes n},
\ket{\theta}_{ABE}^{\otimes(n-r_n)}
\bigr).
\ee
Define
{\bb
\tau^{(n)}\equiv\tau_{A^nB^nR^n}
:=
(I_{A^nB^n}\otimes W^{\otimes n})
\omega_{A^nB^nE^n}
(I_{A^nB^n}\otimes W^{\otimes n})^\dagger.
\ee}
Then \(\tau^{(n)}\) is invariant under simultaneous permutations of the
triples \((A_j,B_j,R_j)\), has marginal \(\rho_{A^nB^n}\) on \(A^nB^n\), and
its support is contained in
\bb
\operatorname{span}{\mathcal V}
\bigl(
\mathcal H_{ABR}^{\otimes n},
\ket{\eta}_{ABR}^{\otimes(n-r_n)}
\bigr).
\ee
Thus \((\tau^{(n)})_n\) is MSR almost-iid along the pure reference state
\(\ket{\eta}\!\bra{\eta}_{ABR}\), with defect sizes \(r_n\).

Applying Lemma~\ref{lem:msr-stability-local-channels}
once more, now to the local tensor-power dephasing channel
\(\Delta_R^{\otimes n}\) on the $R$ systems, we obtain that the sequence
\((\varrho_{R^nA^nB^n})_n\), where
{\bb
\varrho_{R^nA^nB^n}
:=
(\Delta_R^{\otimes n}\otimes\id_{A^nB^n})
(\tau_{A^nB^nR^n}).
\ee}
is MSR almost-iid along \(\varrho_{RAB}\), with defect sizes at most
\(r_n\).
Moreover,
\(\Tr_{R^n}\varrho_{R^nA^nB^n}=\rho_{A^nB^n}\): since \(\Delta_R^{\otimes n}\) is trace-preserving, dephasing the
\(R^n\)-register does not change the marginal on \(A^nB^n\). Hence
\bb
\Tr_{R^n}\varrho_{R^nA^nB^n}
=
\Tr_{R^n}\tau_{R^nA^nB^n}
=
\Tr_{E^n}\omega_{A^nB^nE^n}
=
\rho_{A^nB^n}.
\ee

It remains to refine \(\varrho^{(n)}\) into a cq-extension with pure
conditional states on \(A^nB^n\), at the cost of a subexponential classical
register. Choose a spectral decomposition
\bb
\tau^{(n)}
=
\sum_{\alpha\in\mathcal A_n}
\lambda_\alpha^{(n)}
\ket{\Psi_\alpha^{(n)}}\!\bra{\Psi_\alpha^{(n)}},
\ee
where \(\lambda_\alpha^{(n)}>0\) and the vectors
\(\{\ket{\Psi_\alpha^{(n)}}\}_{\alpha\in\mathcal A_n}\) form an orthonormal
basis of \(\operatorname{supp}\tau^{(n)}\). Since \(\tau^{(n)}\) is supported
on the defect space associated with \(\ket{\eta}_{ABR}\), Lemma~\ref{lem:subexp-defect}
gives
\bb
|\mathcal A_n|
=
\operatorname{rank}\tau^{(n)}
\le
\dim \mathcal M_n(\eta,r_n)
\le
\binom{n}{r_n}d^{r_n},
\qquad
d:=\dim(\mathcal H_A\otimes\mathcal H_B\otimes\mathcal H_R).
\ee
Hence
\bb
\log |\mathcal A_n|
\le
\log\binom{n}{r_n}+r_n\log d
=
o(n),
\ee
because \(r_n=o(n)\).

Let \(\widetilde R_n\) be a classical register with orthonormal basis
\(\{\ket{\alpha}_{\widetilde R_n}:\alpha\in\mathcal A_n\}\). Thus
\bb
\log\dim\widetilde R_n
=
\log|\mathcal A_n|
=
o(n).
\ee
For each \(\alpha\in\mathcal A_n\), write
\bb
\ket{\Psi_\alpha^{(n)}}_{R^nA^nB^n}
=
\sum_{i^n}
\ket{i^n}_{R^n}\otimes
\ket{\psi_{\alpha,i^n}^{(n)}}_{A^nB^n}.
\ee
Define
\bb
q_{\alpha,i^n}^{(n)}
:=
\lambda_\alpha^{(n)}
\bigl\|\psi_{\alpha,i^n}^{(n)}\bigr\|^2,
\ee
and, whenever \(q_{\alpha,i^n}^{(n)}>0\),
\bb
\ket{\phi_{\alpha,i^n}^{(n)}}_{A^nB^n}
:=
\bigl\|\psi_{\alpha,i^n}^{(n)}\bigr\|^{-1}
\ket{\psi_{\alpha,i^n}^{(n)}}_{A^nB^n}.
\ee
Now define
\bb
\widehat\varrho^{(n)}_{\widetilde R_nR^nA^nB^n}
:=
\sum_{\substack{\alpha\in\mathcal A_n,\ i^n:\\
q_{\alpha,i^n}^{(n)}>0}}
q_{\alpha,i^n}^{(n)}
\ket{\alpha}\!\bra{\alpha}_{\widetilde R_n}
\otimes
\ket{i^n}\!\bra{i^n}_{R^n}
\otimes
\ket{\phi_{\alpha,i^n}^{(n)}}\!
\bra{\phi_{\alpha,i^n}^{(n)}}_{A^nB^n}.
\ee
By construction,
\bb
\Tr_{\widetilde R_n}\widehat\varrho^{(n)}_{\widetilde R_nR^nA^nB^n}
=
\varrho^{(n)}_{R^nA^nB^n},
\qquad
\Tr_{\widetilde R_nR^n}\widehat\varrho^{(n)}_{\widetilde R_nR^nA^nB^n}
=
\rho_{A^nB^n}.
\ee
Therefore, \(\widehat\varrho^{(n)}\) is a cq-extension of
\(\rho_{A^nB^n}\) with pure conditional states, the additional
classical register satisfies
\(\log\dim\widetilde R_n=o(n)\), and the marginal sequence
\((\varrho^{(n)}_{R^nA^nB^n})_n\) is MSR almost-iid along
\(\varrho_{RAB}\). This proves the claim.
\end{proof}

\subsection{Proof of Lemma~\ref{lem:MSR-blocking-remainder}}\label{app:msr-block}
\begin{proof}
Let \(n=lm+s\geq m\), with \(l\geq1\) and \(0\leq s<m\),
and fix a purification \(\ket{\theta}_{ABE}\) of \(\rho_{AB}\). By Remark~2.4(a)
of~\cite{mazzola2026almost}, for this purification there exists an
extension \(\omega_{A^nB^nE^n}\) of \(\rho_{A^nB^n}\), invariant under
permutations of the \(n\) triples \((A_j,B_j,E_j)\), such that
\bb
\supp\omega_{A^nB^nE^n}
\subseteq
\operatorname{span}
\mathcal V\!\left(
\mathcal H_{ABE}^{\otimes n},
\ket{\theta}_{ABE}^{\otimes(n-r_n)}
\right).
\ee
We view the first \(lm\) triples as \(l\) consecutive blocks of size
\(m\), and the last \(s\) triples as a remainder.

We claim that the same extension, viewed with this tensor-product
decomposition, satisfies the desired block support condition. It suffices
to check this on the generating vectors of the MSR almost-iid defect space. Such a
generating vector is obtained from
\(\ket{\theta}_{ABE}^{\otimes(n-r_n)}\) by placing an arbitrary vector on
the remaining \(r_n\) one-site tensor factors and then applying a
permutation of the \(n\) sites.

Consider one such generating vector. Among the first \(lm\) tensor
factors, call an \(m\)-block bad if it contains at least one unrestricted
one-site factor. If the number of bad blocks is \(b\), then
\(b\le r_n\) and \(b\le l\). Hence
\bb
b\le t_{l,s}:=\min\{r_n,l\}=\min\{r_{lm+s},l\}.
\ee
Every good \(m\)-block is exactly equal to
\(
\ket{\theta}_{ABE}^{\otimes m}=\ket{\theta_m}_{A^mB^mE^m}
\).
The final \(s\) tensor factors are placed in the unrestricted remainder
system \(\mathcal H_{ABE}^{\otimes s}\). Since \(b\le t_{l,s}\), the
vector has at least \(l-t_{l,s}\) blocks fixed to
\(\ket{\theta_m}\). Equivalently, by allowing additional reference blocks
to be counted among the unrestricted blocks, it lies in
\bb
\operatorname{span}
\mathcal V\!\left(
(\mathcal H_{ABE}^{\otimes m})^{\otimes l},
\ket{\theta_m}^{\otimes(l-t_{l,s})}
\right)
\otimes
\mathcal H_{ABE}^{\otimes s}.
\ee
Taking the span over all generating vectors gives the desired support
inclusion for \(\omega_{A^nB^nE^n}\), viewed as a state on
\((A^mB^mE^m)^lA^sB^sE^s\).

Since \(n=lm+s\), with \(m\) and \(s\) fixed, we have
\bb
\frac{t_{l,s}}{l}
\le
\frac{r_{lm+s}}{l}
=
\frac{r_{lm+s}}{lm+s}\,(m+s/l)
\longrightarrow 0
\quad\text{as }l\to\infty.
\ee
Hence \(t_{l,s}=o(l)\).

Finally, \(\omega_{A^nB^nE^n}\) is invariant under all permutations of the
\(n\) one-site triples. In particular, it is invariant under those
permutations which exchange the \(l\) consecutive blocks of size \(m\)
and leave the remainder fixed. Therefore, viewed as a state on
\((A^mB^mE^m)^lA^sB^sE^s\), it is invariant under permutations of the
\(l\) blocks of size \(m\).
\end{proof}
\subsection{Proof of Corollary~\ref{cor:cq-lifting-remainder}}\label{app:lift-rem}
\begin{proof}
Fix \(m\in\mathbb N_+\) and \(s\in\{0,\ldots,m-1\}\), and write
\(n=lm+s\), with \(l\geq1\). Choose a minimal purification
\(\ket{\theta}_{ABE}\) of \(\rho_{AB}\), with
\(\dim\mathcal H_E=\operatorname{rank}\rho_{AB}\), and set
\bb
\ket{\theta_m}_{A^mB^mE^m}
:=
\ket{\theta}_{ABE}^{\otimes m}.
\ee
By Lemma~\ref{lem:MSR-blocking-remainder}, for each \(l\) there exists an
extension of \(\rho_{A^{lm+s}B^{lm+s}}\) such that
\bb
\operatorname{supp}\omega
\subseteq
\operatorname{span} {\mathcal V}\!\left((\mathcal H_{ABE}^{\otimes m})^{\otimes l},
|\theta_m\rangle^{\otimes(l-t_{l,s})}\right)
\otimes \mathcal H_{ABE}^{\otimes s},
\ee
where \(t_{l,s}=o(l)\), and such that \(\omega\) is invariant under
permutations of the \(l\) blocks of size \(m\).

Let
\bb
\ket{\eta_m}_{A^mB^mR}
:=
\sum_i \sqrt{p_i^{(m)}}\,
\ket{\phi_i^{(m)}}_{A^mB^m}\ket{i}_R .
\ee
{This is a purification of \(\rho_{AB}^{\otimes m}\), while
\(\ket{\theta_m}\) is another purification of the same state. Since
\bb
\dim\mathcal H_E^{\otimes m}
=
\operatorname{rank}(\rho_{AB}^{\otimes m})
\leq
\dim\mathcal H_R,
\ee
there exists an isometry \(W_m:E^m\to R\) such that
\bb
(I_{A^mB^m}\otimes W_m)\ket{\theta_m}
=
\ket{\eta_m}.
\ee
Let \(\Delta_R\) denote the dephasing channel in the basis
\(\{\ket{i}_R\}_i\). Then
\bb
(\Delta_R\otimes\id_{A^mB^m})
\ket{\eta_m}\!\bra{\eta_m}
=
\varrho^{(m)}_{RA^mB^m}.
\ee}
Let \(U_l\) denote the isometry which applies \(W_m\) to the \(E^m\)-part
of each of the \(l\) blocks and acts as the identity on the remaining
systems \(A^sB^sE^s\). Define
\bb
\tau^{(l,s)}
:=
U_l\omega U_l^\dagger .
\ee
By the support inclusion obtained from
Lemma~\ref{lem:MSR-blocking-remainder}, together with the identity
\bb
(I_{A^mB^m}\otimes W_m)\ket{\theta_m}
=
\ket{\eta_m}.
\ee
the state \(\tau^{(l,s)}\) is supported on
\bb
\operatorname{span}
\mathcal V\!\left(
(\mathcal H_{A^mB^mR})^{\otimes l},
\ket{\eta_m}^{\otimes(l-t_{l,s})}
\right)
\otimes
\mathcal H_{ABE}^{\otimes s}.
\ee
Moreover, it is invariant under permutations of the \(l\) blocks. Applying
\(\Delta_R^{\otimes l}\) to the \(R^l\)-system gives
\bb
\zeta^{(l,s)}_{R^l(A^mB^m)^lA^sB^sE^s}
:=
(\Delta_R^{\otimes l}\otimes\id_{(A^mB^m)^lA^sB^sE^s})
(\tau^{(l,s)}).
\ee
Since \(\Delta_R\) is a local channel on the \(R\)-system, the same
local-channel stability argument as in Lemma~\ref{lem:cq-lifting}, now
applied to the \(m\)-blocks, shows that the marginal sequence on
\(R^l(A^mB^m)^l\), obtained by tracing out the bounded remainder
\(A^sB^sE^s\), is MSR almost-iid along
\(\varrho^{(m)}_{RA^mB^m}\), with defect sizes \(t_{l,s}=o(l)\).

{It remains to produce a cq-extension of the full state
\(\rho_{A^{lm+s}B^{lm+s}}\) with pure conditional states and to
bound the additional classical overhead.}
Choose a spectral decomposition
\bb
\tau^{(l,s)}
=
\sum_{\alpha\in\mathcal A_{l,s}}
\lambda_\alpha^{(l,s)}
\ket{\Psi_\alpha^{(l,s)}}\!
\bra{\Psi_\alpha^{(l,s)}},
\ee
with \(\lambda_\alpha^{(l,s)}>0\). Since \(\tau^{(l,s)}\) is supported on
the block-defect space above, tensored with the bounded remainder system,
its rank is at most
\bb
\binom{l}{t_{l,s}}
d_m^{\,t_{l,s}}
d_{\rm rem},
\ee
where
\[
d_m:=\dim(\mathcal H_{A^mB^mR})
\qquad\text{and}\qquad
d_{\rm rem}:=\dim(\mathcal H_{ABE}^{\otimes s}).
\]
Since \(m\) and \(s\) are fixed and \(t_{l,s}=o(l)\), we have
\bb
\log\operatorname{rank}\tau^{(l,s)}
\le
\log\binom{l}{t_{l,s}}
+
t_{l,s}\log d_m
+
\log d_{\rm rem}
=
o(l).
\ee

For each \(\alpha\), expand
\bb
\ket{\Psi_\alpha^{(l,s)}}
=
\sum_{i^l}
\ket{i^l}_{R^l}
\otimes
\ket{\psi_{\alpha,i^l}^{(l,s)}}_{A^{lm+s}B^{lm+s}E^s},
\ee
where the \(E^s\)-system is included in the second tensor factor. Define
\bb
q_{\alpha,i^l}^{(l,s)}
:=
\lambda_\alpha^{(l,s)}
\bigl\|\psi_{\alpha,i^l}^{(l,s)}\bigr\|^2 .
\ee
Whenever \(q_{\alpha,i^l}^{(l,s)}>0\), set
\bb
\ket{\chi_{\alpha,i^l}^{(l,s)}}_{A^{lm+s}B^{lm+s}E^s}
:=
\bigl\|\psi_{\alpha,i^l}^{(l,s)}\bigr\|^{-1}
\ket{\psi_{\alpha,i^l}^{(l,s)}} .
\ee
{The system \(E^s\) is only an auxiliary remainder system in the MSR almost-iid
extension and will not be part of the final cq-extension. Tracing out
\(E^s\), we obtain a state
\bb
\chi_{\alpha,i^l,A^{lm+s}B^{lm+s}}^{(l,s)}
:=
\Tr_{E^s}
\ket{\chi_{\alpha,i^l}^{(l,s)}}\!
\bra{\chi_{\alpha,i^l}^{(l,s)}} .
\ee
{Choose a spectral decomposition into pure states,}
\bb
\chi^{(l,s)}_{\alpha,i^l,A^{lm+s}B^{lm+s}}
=
\sum_j p^{(l,s)}_{j|\alpha,i^l}
\ket{\xi^{(l,s)}_{\alpha,i^l,j}}\!
\bra{\xi^{(l,s)}_{\alpha,i^l,j}}_{A^{lm+s}B^{lm+s}} .
\ee
Since
\bb
\operatorname{rank}
\chi_{\alpha,i^l,A^{lm+s}B^{lm+s}}^{(l,s)}
\le
\dim\mathcal H_E^{\otimes s},
\ee
and \(s<m\) is fixed, this decomposition may be chosen with at most
\(\dim\mathcal H_E^{\otimes s}\) nonzero terms. In particular, the number
of terms is independent of \(l\).

{Absorbing this bounded factor into the classical overhead, let
\(\widetilde R_l\) be a classical register indexing the pairs
\((\alpha,j)\), where \(j\) ranges over a fixed set of size
\(\dim\mathcal H_E^{\otimes s}\), with zero weights used when
necessary. The same \(j\)-labels are used for every \(i^l\). Then
\bb
\log\dim\widetilde R_l
\leq
\log\operatorname{rank}\tau^{(l,s)}
+
\log\dim\mathcal H_E^{\otimes s}
=
o(l).
\ee}

Now define
\bb
\widehat\varrho^{(l,s)}_{\widetilde R_lR^lA^{lm+s}B^{lm+s}}
:=
\sum_{\substack{\alpha,i^l,j:\\
q_{\alpha,i^l}^{(l,s)}
p_{j|\alpha,i^l}^{(l,s)}>0}}
q_{\alpha,i^l}^{(l,s)}
p_{j|\alpha,i^l}^{(l,s)}
\ket{\alpha,j}\!\bra{\alpha,j}_{\widetilde R_l}
\otimes
\ket{i^l}\!\bra{i^l}_{R^l}
\otimes
\ket{\xi_{\alpha,i^l,j}^{(l,s)}}\!
\bra{\xi_{\alpha,i^l,j}^{(l,s)}}_{A^{lm+s}B^{lm+s}} .
\ee
This is a cq-extension of \(\rho_{A^{lm+s}B^{lm+s}}\)
with pure conditional states.
Indeed, tracing over
\(\widetilde R_lR^l\) gives
\bb
\begin{aligned}
\Tr_{\widetilde R_lR^l}
\widehat\varrho^{(l,s)}
&=
\sum_{\alpha,i^l,j}
q_{\alpha,i^l}^{(l,s)}
p_{j|\alpha,i^l}^{(l,s)}
\ket{\xi_{\alpha,i^l,j}^{(l,s)}}\!
\bra{\xi_{\alpha,i^l,j}^{(l,s)}}_{A^{lm+s}B^{lm+s}}
\\
&=
\sum_{\alpha,i^l}
q_{\alpha,i^l}^{(l,s)}
\chi_{\alpha,i^l,A^{lm+s}B^{lm+s}}^{(l,s)}
\\
&=
\Tr_{R^lE^s}\zeta^{(l,s)}
\\
&=
\rho_{A^{lm+s}B^{lm+s}} .
\end{aligned}
\ee
Here the last equality follows because dephasing the \(R^l\)-system does
not change the marginal on \(A^{lm+s}B^{lm+s}E^s\), and
\(\omega\) is an extension of \(\rho_{A^{lm+s}B^{lm+s}}\).

Furthermore, tracing out \(\widetilde R_l\) and the final
\(A^sB^s\)-systems leaves the same marginal on \(R^l(A^mB^m)^l\) as the
one obtained from \(\zeta^{(l,s)}\). Hence this marginal sequence is MSR
almost-iid along \(\varrho^{(m)}_{RA^mB^m}\), with defect sizes
\(t_{l,s}=o(l)\). This proves the corollary.
}

\end{proof}
\end{document}